\documentclass{article}
\usepackage{graphicx} 
\usepackage{amsfonts, amsmath, amssymb, amsthm}
\usepackage{xcolor}
\usepackage{algorithm2e}
\usepackage{bbm}
\usepackage{float}
\usepackage{comment}
\usepackage{subcaption} 
\usepackage{enumerate}
\usepackage{natbib}
\usepackage{enumitem}

\usepackage[letterpaper, margin=1in]{geometry}

\newcommand{\barX}{\overline{X}}
\newcommand{\barx}{\overline{x}}

\newcommand{\TV}{\textrm{TV}}
\newcommand{\KL}{\textrm{KL}}

\DeclareMathOperator{\E}{\mathbb{E}}
\DeclareMathOperator{\Var}{\textrm{Var}}

\newcommand{\ds}{\displaystyle}

\newtheorem{theorem}{Theorem}
\newtheorem{corollary}{Corollary}
\newtheorem{proposition}{Proposition}

\newtheorem{remark}{Remark}

\title{A New Perspective on Reverse Diffusion for Monte Carlo Sampling}

\date{}

\author{
Jairon H. N. Batista \\
FGV EMAp \\
Rio de Janeiro, Brazil \\
\texttt{jairon.batista@fgv.edu.br} \\
\and
Flávio B. Gonçalves \\
UFMG and FGV EMAp \\
Rio de Janeiro, Brazil \\
\texttt{fbgoncalves@ufmg.br} \\
\and
Yuri F. Saporito \\
FGV EMAp \\
Rio de Janeiro, Brazil \\
\texttt{yuri.saporito@fgv.br} \\
\and
Rodrigo S. Targino \\
FGV EMAp \\
Rio de Janeiro, Brazil \\
\texttt{rodrigo.targino@fgv.br}
}

\begin{document}

\RestyleAlgo{ruled}

\maketitle

\begin{abstract}

This paper introduces a novel perspective on the use of reverse diffusion processes for sampling from unnormalized densities. The central idea is to embed the target density as the marginal at the initial time of a suitably constructed diffusion process evolving over a finite horizon. In contrast to existing approaches, the proposed methodology involves neither time discretization error nor score function estimation, so that Monte Carlo variability is the only source of approximation.
A key theoretical result characterizes the Radon-Nikodym derivative of the reverse diffusion transition distribution with respect to that of an Ornstein-Uhlenbeck (OU) process. This representation provides a tractable change-of-measure formulation and serves as the foundation for two distinct classes of Monte Carlo algorithms.
The first class approximates the reverse transition distribution via a sequence of pseudo-marginal Metropolis-Hastings MCMC algorithms. The resulting scheme produces an approximate i.i.d. sample from the target distribution and is fully parallelizable, as trajectories can be generated independently.
The second class consists of MCMC algorithms targeting the joint law of the whole diffusion path in $[0,T]$, for a suitably chosen horizon $T$. The proposed samplers combine three types of updates. One update simulates the diffusion forward in time according to an OU dynamics, conditional on its initial value. The remaining two update the backward component via Metropolis-type steps: one conditions on the terminal value at time $T$ and the other one does not. In both cases, acceptance probabilities are implemented using Barker-type Bernoulli factory constructions.
The proposed methods perform particularly well in the presence of multimodality and remain effective for targets with complex dependence structures, providing a scalable and efficient alternative to the widely used random-walk Metropolis algorithm.

\end{abstract}


\section{Introduction}

Monte Carlo methods have become indispensable tools in modern Bayesian statistics, providing practical means to approximate complex integrals and to explore multi-dimensional posterior distributions that resist analytical treatment. They bridge rigorous theory and applied computation across a wide range of disciplines, including genetics, neuroscience, image reconstruction, spatial analysis, and machine learning \cite{liu2001monte, robert2004monte}. Over the past decades, this has motivated the development of a broad family of simulation-based methods, with particular emphasis on Markov chain Monte Carlo (MCMC) \citep{gamerman2006markov, brooks2011handbook} and Sequential Monte Carlo (SMC) \citep{smcBook, chopin2020introduction} algorithms. While these methodologies have considerably expanded the scope of feasible inference, they still face substantial challenges when the target distribution is highly multimodal or exhibits strong posterior dependence among components -- features that frequently arise in realistic
models encountered in contemporary applications.

Standard Markov chain Monte Carlo algorithms such as Gibbs sampling and random-walk Metropolis-Hastings are well known to perform poorly when the target distribution is multimodal or exhibits strong posterior dependence (ex. highly correlated coordinates). In multimodal settings, local proposals make it difficult for the chain to move between distant modes, leading to metastability and extremely slow mixing times. Likewise, when coordinates are strongly correlated, component-wise Gibbs updates and isotropic random-walk proposals tend to explore the posterior along “wrong” directions, resulting in highly persistent autocorrelations and large effective sample size penalties. To address these issues, the literature has proposed a range of more sophisticated algorithms. Tempering-based methods, such as parallel tempering and simulated tempering, introduce auxiliary temperatures to flatten the target and encourage mode switching, at the cost of running multiple chains and tuning temperature ladders \citep{swendsen1986replica, Geyer1991, neal1996sampling}. Closely related ideas arise in annealing and bridging approaches, including annealed importance sampling and bridge sampling \citep{neal2001annealed, gelman1998simulating}. Population-based Monte Carlo methods provide another strategy by evolving a collection of particles through a sequence of intermediate distributions \citep{cappe2004population}. Sequential Monte Carlo methods offer a natural framework for implementing such schemes, allowing particles to propagate through tempered distributions while incorporating resampling and MCMC mutation steps \citep{del2006sequential}.

Another important class of approaches relies on gradient information to construct more informed proposals. Langevin-based samplers incorporate gradient drift terms to guide proposals toward regions of higher posterior density \citep{roberts1998optimal}, while Hamiltonian Monte Carlo (HMC) generates distant proposals by simulating Hamiltonian dynamics in an augmented phase space \citep{neal2011mcmc}. Subsequent developments include Riemannian manifold variants that adapt to local curvature \citep{girolami2011riemann} and automated tuning strategies such as the No-U-Turn Sampler (NUTS), which dynamically selects integration lengths to avoid inefficient trajectories \citep{hoffman2014no}. These gradient-based methods can dramatically improve exploration in strongly correlated posteriors, but they may still struggle in the presence of well-separated modes and can become computationally expensive in high-dimensional settings due to repeated gradient evaluations and the need to tune numerical integration parameters.

Another line of work investigates non-reversible Monte Carlo methods, which aim at improving sampling efficiency by breaking detailed balance and generating persistent trajectories. In particular, piecewise deterministic Markov processes (PDMPs) \citep{pdmp}, such as the zig-zag process \citep{bierkens2019zig} and the bouncy particle sampler \citep{bouchard2018bouncy} produce continuous-time dynamics with deterministic flows interrupted by random events. These methods can substantially reduce diffusive behavior and improve mixing in some settings, although they typically require access to gradient information and can be challenging to implement for complex models.

More recently, a different line of work has emerged based on denoising diffusion and score-based generative models, in which a forward diffusion process gradually corrupts data with noise, and a learned reverse-time dynamics (often parameterized by neural networks) is used to generate approximate samples from high-dimensional distributions. These methods have achieved remarkable success in producing photorealistic images and other complex data objects in very high dimensions \citep{sohl2015deep, ho2020denoising, song2021score}, and there is growing interest in using them as generic samplers for intractable distributions beyond computer vision. In most of these constructions, the reverse diffusion or score function is only available through an approximation, typically learned from data via deep neural networks, so the resulting algorithms are effectively black-box and their theoretical sampling properties can be difficult to analyze \citep{song2019generative, song2021score}. Moreover, practical implementations invariably rely on time-discretized versions of the reverse stochastic differential equation or its ODE counterpart, which introduce systematic bias unless the discretization is taken to an idealized limit; in addition, learning and evaluating the score network can be computationally demanding. Thus, while diffusion-based samplers offer powerful new mechanisms for exploring complex, high-dimensional distributions, they are almost always approximate, often biased, and not straightforward to adapt when one seeks fully controlled sampling from a prescribed unnormalized target density.

This paper proposes a new framework for sampling from an unnormalized density, such as a Bayesian posterior distribution, through the simulation of reverse diffusion processes. In contrast to diffusion-based generative samplers, the methods developed here do not rely on neural-network approximations of score functions or on time-discretized solutions of stochastic differential equations. As a consequence, the procedures introduced in this work involve only Monte Carlo error, with no additional bias arising from numerical approximations of the underlying diffusions. This shifts the perspective on reverse diffusions from approximate generative modeling to a principled mechanism for exact inference from a prescribed target density.

Our starting point is a diffusion process with the initial distribution given by the target density and whose forward dynamics follows an Ornstein-Uhlenbeck (OU) process converging to a standard Gaussian distribution of matching dimension. The proposed methodology builds on a change-of-measure identity expressing the Radon-Nikodym derivative of the reverse-time dynamics with respect to a non-reverting OU process. Although this representation is exact, its direct use in a rejection sampler is infeasible: the supremum of the Radon-Nikodym derivative is unavailable, and the only attainable upper bound is loose and deteriorates exponentially with time.

To address this, we construct a collection of Monte Carlo building blocks to approximately sample (in Monte Carlo sense) from the reverse transition and marginal distributions of the diffusion whose initial distribution is the target one. Those blocks include sampling-importance-resampling (SIR) and pseudo-marginal Metropolis-Hastings MCMC algorithms.
These blocks are then used to propose two global algorithms targeting the desired distribution. One of them considers independent replications of a sequence of MCMC algorithms sequentially targeting the reverse transition (and optionally the marginal) distribution(s) of the diffusion process, resulting in a Monte Carlo independent sample of the target distribution. The second one involves only one MCMC algorithm targeting the whole diffusion path in a time interval, resulting in an MCMC correlated sample from the target distribution. The second algorithm has higher computational cost but involves the MC error of only one MCMC chain.

Other techniques that use the reverse SDE framework in this context are described in \cite{DiffusionPINN} and  \cite{ReverDiffusionMonteCarlo}. Both works use an Euler–Maruyama time discretization to approximately sample from reverse diffusion, relying on approximations of the score function. In \cite{DiffusionPINN}, the score is approximated by solving a transformed version of the Fokker–Planck PDE using a Physics-Informed Neural Network (PINN), where the initial condition is given by a transformation of $\pi_u$. In contrast, \cite{ReverDiffusionMonteCarlo} estimates the score via a Monte Carlo approach based on Langevin dynamics or importance sampling.
These methods rely on simulating the reverse SDE, which poses a significant challenge due to its complexity. Even in cases where we might have an good approximation of the score function, its behavior is initially unknown, and so we cannot guarantee the stability of numerical discretization methods used for the simulation of the reverse SDE.


\subsection*{Overview of the paper}

The goals and contributions of the paper are the following:
\begin{itemize}
\item introduce a new perspective on the denoising-diffusion approach for sampling from arbitrary unnormalized densities by bridging it with classical SMC and MCMC ideas;

\item propose two classes of algorithms --SPARK and Ping-Pong MCMC-- that operate under this perspective and involve no approximation bias;

\item show how the algorithms arise from two key analytical results for reverse diffusions --a change-of-measure identity for the reverse-time law and an integral identity that enables unbiased estimation of the corresponding transition and marginal densities-- yielding a unified framework for constructing samplers based on reverse stochastic dynamics;

\item demonstrate that the algorithms are \emph{zeroth-order} (requiring no derivative evaluations) and are particularly effective for sampling from \emph{multimodal distributions} and \emph{complex dependence structures};

\item validate the efficiency and robustness of the proposed methods through computational experiments on several benchmark examples.
\end{itemize}

The remainder of the paper is organized as follows. Section~2 presents the theoretical background on reverse diffusions that underpins the proposed methodology. Sections~3 and~4 introduce the two classes of algorithms -- SPARK and Ping-Pong MCMC, respectively, and discusses their main computational aspects related to efficiency, scaling, and practical implementation. Finally, section~5 simulated examples designed to investigate the empirical performance of both algorithms in challenging scenarios.

\section{Theoretical foundation}\label{sec.:Theory}

\subsection{Reverse diffusions}

This section presents the background definitions and results required to develop our methodology. We focus first on the forward diffusion that gradually injects Gaussian noise into an arbitrary initial distribution~$p_0$ and then review the corresponding reverse-time dynamics that play a central role in our algorithms.


Let $(X_s)_{s\ge0}$ solve the $d$-dimensional SDE  
\begin{equation} \label{eq:forward}
dX_s = -\tfrac{1}{2}X_s\, ds + dB_s, 
\qquad X_0 \sim p_0,
\end{equation}
where $(B_s)_{s \ge 0}$ is a $d$-dimensional Brownian motion independent of $X_0$.  
Then $X_s \stackrel{d}{\longrightarrow} \mathcal{N}(0,I_d)$ as $s\to\infty$ (add a reference here).

When $p_0$ is bounded, positive, and $C^2$, \cite{Haussmann1986} showed that the reverse-time process $\barX_t := X_{T-t}$ is itself a diffusion with dynamics
\begin{align} \label{eq:reverse}
d\barX_s 
= \left(\tfrac{1}{2}\barX_s + \nabla \log p(T-s,\barX_s) \right) ds 
+ d\overline{B}_s,
\end{align}
for $\barX_0 = X_T$, where $p(t,\cdot)$ is the density of $X_t$ and $(\overline{B}_s)_{s\ge0}$ is an independent Brownian motion.

We analyze the convergence of the proposed algorithms in terms of both the total variation (TV) distance and the Kullback-Leibler (KL) divergence. Throughout, we treat densities and their associated probability measures interchangeably when evaluating these discrepancies, whenever no ambiguity arises.

The next result quantifies the approximation error of the distribution and moments of $X_t$ to a standard normal.

\begin{proposition}\label{prop_TV}
Let $p_t$ be the distribution of $X_t$ under the SDE in \eqref{eq:forward} and $\gamma_d:=\mathcal{N}(0,I_d)$. Then,
\begin{enumerate}[label=(\roman*)]
\item $\ds \TV(p_T, \gamma_d) \leq \sqrt{\frac{1}{2} \KL(p_T \mid \gamma_d)} \leq e^{-T/2} \sqrt{\frac{1}{2} \KL( p_0 \mid \gamma_d )}$

\item $\ds \mathbb{E}[X_t] = e^{-t/2}\,\mathbb{E}_{p_0}[X]$, \;
$\operatorname{Var}(X_{t,i}) = 1 + e^{-t}\left(\operatorname{Var}_{p_0}(X_i)-1\right)$, \;
$\operatorname{Cov}(X_{t,i},X_{t,j}) = e^{-t} \operatorname{Cov}_{p_0}(X_i,X_j)$, 
 $i\neq j$, provided that the corresponding moments under $p_0$ are finite.
\end{enumerate}
\end{proposition}



\subsection{Main theoretical results}

The methodology developed in this paper is built upon two key theoretical results.
The first result is a change-of-measure formula giving the Radon-Nikodym derivative between the law of the reverse diffusion in~\eqref{eq:reverse} and that of a non-reverting OU process.
It is obtained as a corollary of a more general theorem that establishes the analogous Radon-Nikodym derivative for the reverse diffusion associated with a broad class of forward diffusions.
The second result is an integral identity yielding an unbiased estimator of the marginal density $p(t,x)$ of the forward diffusion with initial distribution $p_0$ and OU dynamics.

\subsubsection{Girsanov change of measure for the reverse diffusion}


\begin{theorem} \label{th:girsanovgeral}
Let $p_0$ be a $C^2$ density on $\mathbb{R}^d$, and let 
$f:\mathbb{R}^d\times[0,T]\to\mathbb{R}^d$ 
and 
$g:[0,T]\to\mathbb{R}_{>0}$ 
satisfy the regularity assumptions guaranteeing a unique weak solution of
\[
dX_s = f(s,X_s)ds + g(s)dB_s, 
\qquad X_0\sim p_0,
\]
with marginal density $p(t,\cdot)$ of $X_t$, which is positive and $C^2$, solving the Fokker-Planck equation
\[
\partial_t p(t,x) 
= -\nabla\cdot[f(t,x)p(t,x)] 
+ \tfrac{1}{2}\Delta p(t,x).
\]
Let $C := C([0,T],\mathbb{R})$ be the space of continuous paths, with cylinder $\sigma$-algebra~$\mathcal{C}$.  
For $0<s_0<s_1<T$, let $\bar P$ and $\bar Q$ be respectively the path measures on $(C,\mathcal{C})$ (restricted to $[s_0,s_1]$) induced by the solutions of
\begin{align*}
d\barX_s 
&= \left(-f(T-s,\barX_s) + g(T-s)^2\,\nabla\log p(T-s,\barX_s)\right)ds 
+ g(T-s)d\overline{B}_s, \\
d\overline{Y}_s 
&= -f(T-s,\overline{Y}_s)ds + g(T-s)d\overline{B}_s,
\end{align*}
with common initial condition $\barX_{s_0}=\overline{Y}_{s_0}=\barx_0$.  
Then,
\[
\frac{d\bar P}{d\bar Q}(\omega)
= \exp\!\left\{-\int_{s_0}^{s_1} \nabla\cdot f(T-s,\omega_s)\,ds \right\}
\frac{p(T-s_1,\omega_{s_1})}{p(T-s_0,\barx_0)},
\]
provided the right-hand side is a true $\bar Q$-martingale, which is guaranteed when $\nabla\cdot f$ and $p$ are bounded.
\end{theorem}


\begin{corollary}\label{Corl1}
Suppose that $p_0$ is bounded and $C^2$ and take $f(s,x)=-\tfrac{1}{2}x$ and $g(s)=1$. Then 
$\overline{Y}_{s+t}\mid \overline{Y}_s=x\sim \mathcal{N}(e^{t/2}x,(e^t-1)I_d)$ 
and
\begin{equation}\label{eq:RNrev_OU}
\frac{d\bar P}{d\bar Q}(\omega)
= \exp\!\left\{\tfrac{d}{2}(s_1-s_0)\right\}
\frac{p(T-s_1,\omega_{s_1})}{p(T-s_0,\barx_0)}.
\end{equation}
This implies that:
\begin{enumerate}[label=(\roman*)]
\item conditional on $\omega_{s_1}$, $\bar P \equiv \bar Q$;
\item the ratio of the transition densities of $\omega_{s_1}\mid\omega_{s_0}$ under $\bar P$ and $\bar Q$ equals the expression in \eqref{eq:RNrev_OU}.
\end{enumerate}
\end{corollary}

The next proposition, a direct consequence of the previous result, characterizes the target density of the MCMC algorithms proposed in the paper.

\begin{proposition}\label{prop:targt_mcmc}
Consider the definitions of $f$ and $g$ in Corollary \ref{Corl1}. Define $\omega_{0:T}=\{\omega_t,\;t\in[0,T]\}$ and let $\bar{q}$ be the Lebesgue density of $\omega_0$ under $\bar{Q}$. Then,
\begin{equation}\label{eq:RNrev_OU_times}
\frac{d\bar P}{d\bar Q}(\omega_{0:J}) 
\propto \frac{p_0(\omega_T)}{q(\omega_0)}.
\end{equation}
\end{proposition}

\subsubsection{Unbiased estimation of the marginal density of the forward diffusion}\label{ssec_IS}

Assume that the initial density has the form \(p_0(x)=c\,\pi_u(x)\), where \(\pi_u\) is a known unnormalized density.

\begin{proposition}\label{targino}
Let \(p(t,x)\) denote the marginal density of \(X_t\) from~\eqref{eq:forward}. Then
\begin{equation}\label{eq:unb_est0}
p(t,x)
= c\,e^{\frac{d}{2}t}\,
\mathbb{E}_{Z_{t,x}}\!\left[\pi_u(Z_{t,x})\right],\quad Z_{t,x}\sim\mathcal{N}(e^{t/2}x,(e^t-1)I_d).
\end{equation}
\end{proposition}

This yields the unbiased estimator $\hat p(t,x)= c\,e^{\frac{d}{2}t}\,\frac{1}{m}\sum_{i=1}^m
\pi_u(Z_i)$, which evaluates \(\pi_u\) at points of order \(\mathcal{O}(e^{t/2})\) (for \(x=\mathcal{O}(1)\)), and thus becomes increasingly unstable as \(t\) grows, since it probes the tails of \(\pi\). To stabilize this behavior, we introduce an importance sampling (IS) version of the estimator.

Let \(\tilde q_t\) denote the density of \(\mathcal{N}(e^{t/2}x,(e^t-1)I_d)\). For a proposal density \(q_t\), define
\begin{equation}\label{eq:unb_est2}
\hat p_{\mathrm{IS}}(t,x)
=
c\,e^{\frac{d}{2}t}\,
\frac{1}{m}\sum_{i=1}^m
\pi_u(Z_i)\,w(Z_i,t),
\qquad Z_i\stackrel{iid}{\sim}q_t, \quad w(z,t)=\frac{\tilde q_t(z)}{q_t(z)}.
\end{equation}
Since only an unbiased estimator of the unnormalized $p(t,\cdot)$ is needed withing the proposed algorithms, the constant term $c\,e^{\frac{d}{2}t}/m$ can be ignored.

If \(\pi=\mathcal{N}(\mu,\sigma^2 I_d)\), then the optimal (zero-variance) IS proposal is proportional to \(\pi(y)\tilde q_t(y)\), and is therefore Gaussian:
\begin{equation}\label{IS_prop}
q_t^G
=
\mathcal{N}\!\left(
\frac{(e^t-1)\mu + \sigma^2 e^{t/2}x}{\sigma^2+e^t-1},
\;
\frac{\sigma^2(e^t-1)}{\sigma^2+e^t-1}I_d
\right).    
\end{equation}
This interpolates between a point mass at \(x\) as \(t\downarrow 0\) and \(\pi\) as \(t\to\infty\).

When \(\pi\) is not Gaussian, we mimic this optimal structure and make $q_t=q_t^G$. In practice, \(\pi\) is standardized (e.g., via preliminary MCMC runs), and we set \(\mu=0\) and \(\sigma^2>1\). The choice \(\sigma^2>1\) inflates the proposal variance, providing robustness against multimodality and non-linear dependence of $\pi_u$.

The IS estimator may exhibit large or even infinite variance if the proposal \(q_t\) has lighter tails than \(\tilde q_t\). Its variance is finite if and only if
\[
\mathbb{E}_{q_t}\!\left[\pi_u(Z)^2 w(Z,t)^2\right] < \infty.
\]
Since \(\pi_u\) is assumed bounded, we only require $\mathbb{E}_{q_t}[w(Z,t)^2] < \infty$.

For the Gaussian proposal described above, this condition holds if and only if $\sigma^2 > e^t - 1$.
However, this requirement is of limited practical interest, as it forces the proposal variance to grow at the same rate as that of \(\tilde q_t\), and therefore does not achieve the goal of stabilizing the argument of \(\pi\) in the IS estimator.

If \(\pi\) has (sub-)Gaussian tails, the condition can be substantially relaxed, as shown in the following result.

\begin{proposition}\label{prop.GTEst}
Suppose that there exist constants \(C>0\), \(a>0\), and \(R<\infty\) such that
\[
\pi(y)\leq C e^{-a\|y\|^2}, \quad \|y\|>R.
\]
Then the IS estimator in \eqref{eq:unb_est2} has finite variance if $\ds \sigma^2>\frac{e^t-1}{1+4a(e^t-1)}$.
\end{proposition}

More generally, variance can be expressed in terms of the chi-square divergence.

\begin{proposition}\label{prop.stTEst}
Let
\[
q^\star(y)=\frac{\pi(y)\tilde q_t(y)}{r},
\qquad
r=\int \pi(y)\tilde q_t(y)\,dy.
\]
Then the single-sample IS estimator of \(r\) has variance
\[
\operatorname{Var}_{q_t}\!\left(\frac{\pi(Y)\tilde q_t(Y)}{q_t(Y)}\right)
=
r^2\,\chi^2(q^\star\|q_t), \quad
\text{for}\quad
\chi^2(q^\star\|q_t)
=
\mathbb{E}_{q_t}\!\left[\left(\frac{q^\star}{q_t}(Y)-1\right)^2\right].
\]
\end{proposition}

\medskip

\noindent
The above result provides a clear criterion for assessing the effectiveness of the IS proposal. In particular, the variance of the IS estimator is reduced relative to the original estimator if and only if
\[
\chi^2(q^{\star}\|q_t) < \chi^2(q^{\star}\|\tilde q_t)=\mathrm{CV}_{\tilde q_t}^2(\pi).
\]
That is, variance reduction is achieved precisely when the proposal \(q_t\) is closer to the optimal distribution \(q^{\star}\) than the original sampling distribution \(\tilde q_t\), in the sense of the chi-square divergence.

\medskip

\noindent
\textbf{Heavy-tailed targets.}
We now construct a robust IS proposal when \(\pi\) has (sub-)Student-\(t\) tails. Suppose that there exist constants \(C>0\), \(\nu>0\), and \(R<\infty\) such that
\[
\pi(y)\leq C(1+\|y\|^2)^{-(\nu+d)/2}, \quad \|y\|>R.
\]

In this setting, the optimal proposal is not available in closed form. We retain the Gaussian-optimal location and scale, and replace the Gaussian by a Student-\(t\) distribution:
\[
q_t(y)
=
t_{\nu_q}\left(
y;
\frac{\sigma^2 e^{t/2}x}{\sigma^2+e^t-1},
\;
\omega_t^2 I_d
\right),\qquad
\omega_t^2
=
\frac{\sigma^2(e^t-1)}{\sigma^2+e^t-1}\,c(\nu_q),
\quad
c(\nu_q)=
\begin{cases}
\frac{\nu_q-2}{\nu_q}, & \nu_q>2,\\
1, & \nu_q\leq 2.
\end{cases}
\]

The resulting IS estimator has finite variance for all \(\nu>0\) and \(\nu_q>0\), due to the Gaussian factor in \(\tilde q_t\).

In practice, we recommend choosing $\nu_q \lesssim \nu$  and $\sigma^2>1$, so that the proposal has tails at least as heavy as those of \(\pi\), while retaining efficiency.

\subsection{Monte Carlo building blocks}\label{ssec_MCBB}

The IS estimator devised in the previous section enables unbiased estimation not only of the marginal density of \(X_t\), but also of the transition density of the reverse diffusion, via Corollary~\ref{Corl1}. These estimators, in turn, provide the building blocks for (approximate) sampling from these distributions using standard Monte Carlo algorithms.

Classical Monte Carlo methods such as rejection sampling, sampling--importance--resampling (SIR), and Metropolis-type MCMC depend on the ratio between target and proposal densities. These methods can still be implemented when the target density is intractable, provided that an almost surely positive unbiased estimator is available.

In rejection sampling, if a proposal value \(x\) has acceptance probability \(\alpha(x)\), it suffices to construct an unbiased estimator \(\hat{\alpha}(x,U)\in[0,1]\) such that \(\mathbb{E}_U[\hat{\alpha}(x,U)] = \alpha(x)\), and then simulate a Bernoulli random variable with success probability \(\hat{\alpha}(x,U)\) to determine acceptance. In the SIR algorithm, the unnormalized importance weights can be replaced by unbiased estimators. Similarly, pseudo-marginal Metropolis--Hastings (PMMH) can be implemented on an augmented state space by replacing the target density with its unbiased estimator \citep{andrieu2009}.

In the present setting, however, the uniform bound on \(\pi\) required for rejection sampling leads to a loose control of the density ratio. As a consequence, the global acceptance probability decays at least exponentially fast in \(t\) (see Proposition \ref{prop:RSacc} in Appendix \ref{ap:EA}), rendering the method impractical for moderate or large values of \(t\).

By contrast, SIR and PMMH do not require such a bound, at the cost of only approximately sampling from the target distribution. A key difference between them is that SIR relies on independent proposals, whereas PMMH allows for local (e.g., random walk-type) proposals. Since no natural independent proposal is available for the marginal of \(X_t\), and the Gaussian proposal \(\mathcal{N}(e^{\Delta/2}x,(e^\Delta-1)I_d)\) for the reverse transition over an interval of length \(\Delta\) may be poorly matched to the target, PMMH with local proposals is preferred. Nevertheless, SIR algorithms targeting the marginal of $X_T$ and the reverse transition remain of interest, as they enable the construction of global MCMC schemes on \(X_0\). We therefore consider SIR for $X_T$ and for the reverse transition, as well as two PMMH schemes: one targeting the reverse transition and one targeting the marginal of \(X_t\).

\begin{remark}
A standard particle filter could be constructed by propagating and resampling \(N\) particles via SIR. However, the weights at time \(t\) involve the intractable quantity \(p(t+\Delta,\cdot)\). In the single-particle scheme, this term is common across candidates and cancels in the normalized weights. In a multi-particle setting, it varies across particles and cannot be ignored, requiring unbiased estimators of both \(p(t+\Delta,\cdot)\) and its reciprocal \citep{turnkey2025}, with a substantial increase in complexity.
\end{remark}

\paragraph{SIR for the reverse transition.}

Let \(0 \leq t < t+\Delta \leq T\). By Corollary~\ref{Corl1}, the conditional distribution of \(X_t \mid X_{t+\Delta} = x\) under~\eqref{eq:forward} is absolutely continuous with respect to
\[
\mathcal{N}\!\left(e^{\Delta/2}x,(e^{\Delta}-1)I_d\right),
\]
with Radon--Nikodym derivative proportional to \(p(t, \cdot)\).

Since \(p(t,\cdot)\) is not available in closed form (except at \(t=0\)), we replace this quantity by the unbiased estimator in~\eqref{eq:unb_est2}. This leads to the following SIR algorithm.

\begin{algorithm}[H]
\caption{SIR algorithm}\label{alg:sir}
\KwData{Time $t$ and increment $\Delta$, number of particles $n$ and estimator sample size $m$.}
\KwResult{An approximate realization of $X_t \mid X_{t+\Delta}=x$}

Draw $X_t^{(i)} \sim 
\mathcal{N}\!\left(e^{\Delta/2}x,(e^{\Delta}-1)I_d\right)$,
$i = 1,\dots,n$\;

Draw $Z_{\ell}^{(i)} \sim q_t$,
$\ell = 1,\dots,m$, $i = 1,\dots,n$\;

$\ds w^{(i)} \gets 
\hat p_{\mathrm{IS}}(t,X_{t}^{(i)}),
\quad \text{if } t>0,\qquad
w^{(i)} \gets \pi_u\!\left(X_0^{(i)}\right), \ \text{if } t=0$\;

Normalize weights: $\ds W^{(i)} \gets
\frac{w^{(i)}}{\sum_{k=1}^{n} w^{(k)}}$\;

Draw $I \sim \mathrm{Categorical}\!\left(
W^{(1)},\dots,W^{(n)}\right)$\;

Set $X_t \gets X_t^{(I)}$\;
\end{algorithm}

\paragraph{SIR for the marginal of $X_T$.}

The SIR algorithm for $X_T$ draws $n_J$ particles from $\bar{q}_T$ and selects one with probability proportional to an unbiased estimator of the weight $p(T,\cdot)/\bar{q}_T(\cdot)$, where the numerator is replaced by the estimator in~\eqref{eq:unb_est2} with $t=T$.

\paragraph{PMMH for the reverse transition.}

A key feature of the reverse-time formulation is that the conditional distribution of \(X_t \mid X_{t+\Delta}=x\) admits a natural Gaussian reference measure. Specifically, it has a density with respect to
$\ds
\mathcal{N}\!\left(e^{\Delta/2}x,(e^{\Delta}-1)I_d\right)$,
proportional to \(p(t,\cdot)\). This structure naturally motivates the use of Gaussian-reference MCMC methods, such as the preconditioned Crank--Nicolson (pCN) algorithm, whose proposals preserve the reference measure exactly.

In contrast to standard random-walk proposals, whose efficiency typically deteriorates with dimension due to the accumulation of independent perturbations, pCN proposals respect the underlying Gaussian geometry and correct only through the Radon--Nikodym derivative. This advantage becomes more pronounced as \(t\) increases: the OU dynamics smooth the marginal density \(p_t\), leading to a more regular Radon--Nikodym derivative and improved mixing.

In this setting, the acceptance probability of a move \(y \to y^*\) is
\begin{equation}\label{PMMH_ap}
1 \wedge \frac{p(t,y^*)}{p(t,y)},
\qquad
\text{or}
\qquad
1 \wedge \frac{\hat{p}_{\mathrm{IS}}(t,y^*)}{\hat{p}_{\mathrm{IS}}(t,y)}
\end{equation}
for the ideal and pseudo-marginal cases, respectively.

The pCN proposal is given by
\begin{equation}
y^* = e^{\Delta/2}x + \sqrt{1-\beta^2}\,(y - e^{\Delta/2}x) + \beta (e^{\Delta}-1)^{1/2} Z,
\qquad Z \sim \mathcal{N}(0,I_d),
\end{equation}
where \(\beta \in (0,1)\) is tuned to achieve suitable acceptance rates.

\paragraph{PMMH for the marginal of \(X_t\).}

When targeting the marginal distribution of \(X_t\), no such Gaussian reference structure is available. In this case, we use a standard random-walk proposal,
\[
y^* \sim \mathcal{N}(y,\sigma^2 I_d),
\]
with \(\sigma^2\) tuned according to standard optimal scaling considerations \citep{roberts1997weak}. The corresponding acceptance probabilities are again given by~\eqref{PMMH_ap}.

Finally, note that the chains targeting \(X_0 \mid X_{\Delta}\) and the marginal \(X_0\) are standard Metropolis--Hastings algorithms, since the corresponding unnormalized target densities are tractable.

\section{The SPARK algorithm}

We introduce two algorithms for generating approximate samples from the target distribution \(p_0\), the marginal law of \(X_0\). We call the algorithms \emph{Sequential Pseudo-marginal Algorithm with Reverse-time Kernels} (SPARK). Both are constructed as sequences of MCMC schemes along a time schedule
\[
0 = t_0 < t_1 < \cdots < t_J = T.
\]

The construction is guided by two principles: (i) the sequence of target distributions evolves smoothly in time, so that successive targets are locally close; and (ii) each MCMC step is required only to mix locally around its initialization, which is inherited from the previous step. This construction can be viewed as a sequential local approximation to a global MCMC targeting \(p_0\).

Each run of the algorithm produces a single approximate draw from \(p_0\); independent repetitions yield an approximate i.i.d.\ sample.

The first algorithm consists of a sequence of PMMH chains targeting
\[
X_T, \quad X_{t_{J-1}} \mid X_T, \quad \ldots, \quad X_{t_1} \mid X_{t_2}, \quad X_0 \mid X_{t_1},
\]
where the output of each chain provides the conditioning value for the next. Since the targets evolve gradually, each chain need only explore a local neighborhood of its starting point. The procedure yields an approximate reverse-time trajectory and, in particular, a draw from \(X_0\).

The second algorithm augments this sequence by inserting PMMH steps targeting the marginals:
\[
X_T,\quad
\{(X_{t_{j}} \mid X_{t_{j+1}}),\;\; X_{t_j}\}_{j=J-1}^{1},\quad
X_0 \mid X_{t_1},\;\; X_0.
\]
These additional steps mitigate the accumulation of approximation error by periodically re-targeting the marginal distributions. The resulting scheme no longer produces trajectories, but yields approximate samples from each marginal; in particular, the final step is a standard Metropolis--Hastings algorithm targeting \(p_0\).

The marginal-refreshing steps increase computational cost, but can substantially reduce the effort required in the conditional chains. In practice, this trade-off often favors the second algorithm.

\begin{algorithm}[H]
\caption{SPARK algorithm}\label{alg:spark}
\KwData{Time schedule $0 = t_0 < \cdots < t_J = T$; iterations $n_J$, $\{n_{j1},n_{j2}\}_{j=0}^{J-1}$; PM estimator sample sizes; number of replications $M$.}
\KwResult{Independent approximate samples $\{x_0^{(i)}\}_{i=1}^M$}

\For{$i=1,\dots,M$}{
    Run a PMMH chain targeting $X_T$ for $n_J$ iterations; set $x_{t_J}^{(i)}$\;
    
    \For{$j=J-1,\dots,0$}{
        Run a PMMH chain (MH if $j=0$) targeting $X_{t_j}\mid X_{t_{j+1}}=x_{t_{j+1}}^{(i)}$ for $n_{j1}$ iterations; set $\tilde{x}_{t_j}^{(i)}$\;
        
        Run a PMMH chain (MH if $j=0$) targeting $X_{t_j}$ initialized at $\tilde{x}_{t_j}^{(i)}$ for $n_{j2}$ iterations; set $x_{t_j}^{(i)}$\;
    }
}
\Return{$\{x_0^{(i)}\}_{i=1}^M$}
\end{algorithm}

Unlike standard Metropolis--Hastings algorithms, for which a well-developed asymptotic theory provides robust and widely applicable guidelines for tuning proposal scales and acceptance rates, there is comparatively limited theory for tuning the variance of the estimator and the resulting acceptance rate in pseudo-marginal methods. The main available results, notably \citet{SherlockThieryRobertsRosenthal2015} and \citet{DoucetPittDeligiannidisKohn2015}, derive optimality criteria under a number of strong and stylized assumptions. These typically include: (i) an additive noise representation for the log-target estimator; (ii) independence between the noise distribution and the point at which the target is evaluated; (iii) approximate Gaussianity of the log-noise, with variance inversely proportional to the computational effort; and, in some cases, (iv) high-dimensional asymptotics with product-form targets and isotropic random-walk proposals. While analytically convenient, these assumptions are unlikely to hold in many realistic applications and are explicitly acknowledged as approximations in the aforementioned works. In particular, the independence of the noise from the state is generally violated in practice and is introduced primarily to enable tractable analysis.

Moreover, the structure considered here differs substantially from the setting of these results. We employ a sequence of PMMH kernels targeting a collection of intermediate distributions, with only the final target corresponding to the posterior of interest; earlier targets are close to a Gaussian reference. In addition, the unnormalized target density is available in closed form, so that the final MH step is exact rather than pseudo-marginal, and each chain is run only for a small number of iterations with the aim of producing a single approximately independent draw via multiple independent replications. This regime emphasizes local mixing rather than long-run stationary efficiency and falls outside the asymptotic frameworks considered in the existing theory.

Taken together, these considerations suggest that the available prescriptions, such as targeting a specific variance of the log-target estimator or a corresponding acceptance rate, should not be expected to provide reliable or universal guidance in this setting. Instead, we advocate tuning the variance and acceptance rate empirically, on a case-by-case basis, using diagnostics that directly reflect computational efficiency, such as effective sample size per unit time, acceptance probabilities, short-lag autocorrelations, and the stability of the estimator across the relevant regions of the state space.

\subsection{Practical considerations}

\paragraph{Rescaling $p_0$.}

Since the Ornstein--Uhlenbeck (OU) forward diffusion converges to a standard Gaussian distribution regardless of the initial law, it is natural to place the initial target on the same global scale as this reference measure. Accordingly, we standardize $p_0$ to have mean zero and unit variance using model knowledge and/or empirical estimates. This normalization also facilitates the specification of proposal distributions in the importance sampling estimators of $p_t$ (see Section~\ref{ssec_IS}).

When a reliable estimate of the correlation structure is available, we further whiten the target so that its covariance matrix is close to the identity. This removes linear dependence and isolates the genuinely non-Gaussian features of the distribution, such as multimodality, skewness, non-linear dependence, or heavy tails. The OU dynamics then act primarily on these features, progressively regularizing the distribution as $t$ increases.

From an algorithmic perspective, this normalization provides a convenient baseline under which the time horizon $T$ directly controls the degree of smoothing required for the marginal $X_T$ to become amenable to efficient local MCMC exploration.

In practice, the required transformations can be obtained from short preliminary runs, as is standard in algorithms such as random-walk Metropolis--Hastings. The resulting estimates need not be highly accurate: their role is simply to place the problem on a sensible global scale and reduce linear dependence prior to applying SPARK.

\paragraph{Choice of the terminal time $T$.}

The terminal time $T$ plays a central role in the efficiency of SPARK, as it governs a fundamental trade-off. As $T$ increases, the marginal distribution of $X_T$ approaches the standard Gaussian, making it easier for independent replications of the PMMH algorithm targeting this marginal to approximate it accurately. Conversely, larger values of $T$ require finer time discretizations (larger $J$), increasing the overall computational cost.

A natural strategy is therefore to select the smallest value of $T$ for which independent replications of a relatively inexpensive PMMH scheme targeting $p_T$ yield a representative sample. This choice can be guided by theoretical results such as the ones in Proposition~\ref{prop_TV}.

\paragraph{Choice of the time schedule.}

The time schedule $\{t_j\}_{j=0}^J$ is also critical for efficiency. There is again a trade-off: coarser schedules (smaller $J$) reduce the number of PMMH chains but make the corresponding conditional targets more challenging, while finer schedules improve local mixing at the cost of additional chains.

In line with the choice of $T$, one should therefore select the coarsest schedule for which the conditional PMMH updates remain reasonably efficient. More importantly, the placement of the time points should reflect the non-uniform rate at which the diffusion regularizes the target. In particular, Proposition~\ref{prop_TV} indicates that convergence is exponential in $t$, suggesting that the spacing between consecutive times should decrease as $t \downarrow 0$.

A practical way to balance efficiency across levels is to choose the schedule so that the conditional PMMH kernels exhibit comparable acceptance rates for a common choice of the pCN tuning parameter $\beta$.

\paragraph{Specification of the PMMH algorithms.}

We now discuss the main specifications of the PMMH algorithms within SPARK, including initialization, estimator variance, acceptance rates, and number of iterations.

Initialization is straightforward. For the marginal chain at $t=T$, it is natural to draw the initial value from $\mathcal{N}(0,I_d)$, as the target is expected to be close to this distribution. For chains targeting $X_{t_j}$ with $t_j<T$, a natural choice is to initialize at the final state of the preceding conditional PMMH chain targeting $X_{t_j}\mid X_{t_{j+1}}$. For this conditional chain, a convenient initialization is to draw a random value from the Gaussian reference distribution $\mathcal{N}\!\big(e^{\Delta_j}x_{t_{j+1}}, (e^{\Delta_j}-1)I_d\big)$, for $\Delta_j=t_{j+1}-t_j$.

The variance of the importance sampling estimator in \eqref{eq:unb_est2} typically increases with $t$, implying that the corresponding pseudo-marginal noise is larger for early stages of the algorithm. Empirical estimates of the variance of the log-noise can be obtained at representative points (e.g.\ within $[-2,2]^d$ under the standardized scale). In practice, however, the most reliable guidance comes from standard MCMC diagnostics, including trace plots, autocorrelation functions, and the dependence of the acceptance rate on the proposal tuning.

As $t \downarrow 0$, the variance of the estimator decreases and the PMMH chains approach the ideal marginal MH algorithm. Accordingly, the optimal acceptance rate is expected to increase toward the standard MH regime (e.g.\ $0.44$ in dimension $1$ and $0.234$ in moderate-to-high dimensions), while for larger $t$ lower acceptance rates are typically appropriate. As a rough reference, \citet{SherlockThieryRobertsRosenthal2015} suggest an acceptance rate of approximately $0.07$ for random-walk PMMH in high dimensions.

Given these considerations, and since each chain is only required to mix locally, a few hundred iterations are typically sufficient for each PMMH component.

\paragraph{Alleviating the computational cost.}

A potential limitation of SPARK is its computational cost, as it requires multiple independent replications of a sequence of MCMC chains, each involving repeated evaluations of the unnormalized density $\pi_u$.

On the other hand, the independence of the replications yields an independent sample, eliminating the long-range autocorrelation typically associated with standard MCMC. This structure can be exploited to reduce the number of required replications. In particular, the last chain in the SPARK sequence is a standard MH chain targeting the marginal of $X_0$; instead of running this chain only long enough to produce a single output per replication, one may run it for a larger number of iterations and retain multiple samples. Since this chain operates in a regime where only local mixing is required, it can efficiently explore local structure, thereby reducing the burden on the replication mechanism.

For example, in multimodal settings, the independent replications primarily serve to approximate the relative weights of the modes, while the local MH chains targeting $X_0$ are responsible for exploring each mode. This separation of global and local mixing can substantially reduce the overall computational cost.

\section{Ping-Pong MCMC}\label{sec.:pingpong}

The SPARK algorithm is typically the more robust option in difficult examples, since it relies on a sequence of MCMC updates with random-walk type proposals rather than on sequential importance-resampling steps with independent proposals. Its limitation is that the final approximation depends on the joint behaviour of several Markov chains. The \emph{Ping-Pong} construction provides a complementary alternative: it defines a single Markov chain whose marginal invariant distribution is exactly $p_0$. It is therefore less generally robust than SPARK, but can offer higher precision when the corresponding proposals are sufficiently accurate, and it can also be embedded as an exact block update within larger Gibbs sampling schemes, with $p_0$ playing the role of a full conditional density.

A direct Metropolis--Hastings correction of SPARK itself would require the density, or an unbiased estimator of the density, of its output distribution. This is not available because the SPARK proposal is generated through internal accept--reject MCMC dynamics. We therefore replace SPARK by a sequential reverse-time SIR proposal, for which unbiased estimators of the output density are available. We call this proposal mechanism \emph{Sequential Path Importance sampling for Diffusion Equations in Reverse time} (SPIDER).

\begin{algorithm}[H]
\caption{SPIDER algorithm}\label{alg:spider}
\KwData{Time schedule $0=t_0<\cdots<t_J=T$; particles $\{N_j\}_{j=0}^J$; weight estimator sample sizes.}
\KwResult{An approximate path sample $\{x_{t_j}\}_{j=0}^J$}

Sample $x_T$ by SIR with proposal $\bar{q}_T$ and weights $\hat p_{IS}(T,\cdot)/\bar{q}_T(\cdot)$\;
\For{$j=J-1,\ldots,0$}{
    Conditional on $x_{t_{j+1}}$, sample $x_{t_j}$ by SIR using Algorithm~\ref{alg:sir}, with $n=N_j$\;
}
\Return{$\{x_{t_j}\}_{j=0}^J$}
\end{algorithm}

The Ping-Pong algorithms operate (initially) on the extended path $X_0,\;X_{t_1},\ldots,X_{t_{J-1}},\;X_T$,
and use SPIDER proposals inside Metropolis-type transitions. Three elementary updates are used: 1) a forward update of $X_{t_1},\ldots,X_T\mid X_0$, simulated exactly from~\eqref{eq:forward}; 2) a backward update of $X_0,\ldots,X_{t_{J-1}}\mid X_T$, proposed by SPIDER and corrected by Barker's rule; and 3) a joint update at the whole time schedule, also proposed by SPIDER and corrected by Barker's rule. Alternating updates 1) and 2) yields a Gibbs-type scheme, update 3) alone gives an independence MH-type chain, and hybrid schemes combine local correction with occasional global moves.

The algorithms can be formulated on the space of continuous trajectories in $[0,T]$. By Corollary~\ref{Corl1}, the bridges of $\bar{P}$ and $\bar{Q}$ have the same law, so only the values at the selected time schedule enter the accept--reject step explicitly. Consequently, the time schedule may be changed across MCMC iterations. In practice, one may keep comparable increments while shifting the time points; since both proposal and acceptance probabilities depend on the schedule, this can help avoid persistent sticking caused by an unfavourable fixed schedule. Intermediate path values can be sampled retrospectively, after the accept-reject step, from the corresponding OU bridge law.

The SPIDER output density is intractable, but it admits a positive bounded unbiased estimator. Let $\tilde Q$ denote the SPIDER law, completed with bridges from $\bar Q$ and assume that $X_T$ has Lebesgue density $\bar{q}_T$ under $\bar{Q}$, we have that (see Proposition \ref{prop.ueSP_IS} in Appendix \ref{ssec.:MCMCanls})
\[
q_S(x):=\frac{d\tilde{Q}}{d\bar{Q}}(x)=q_J(x_T) \prod_{j=J-1}^0 q_j(x_{t_{j}}|x_{t_{j+1}}),
\]
\begin{equation}\label{eq.qjSP}
q_j(x_{t_{j}}|x_{t_{j+1}})
=
n_j
\mathbb{E}_{\dot{X}_j,Z_j}
\left[
\frac{
w_{j}(x_{t_{j}},Z_{j,n_j})
}{
w_{j}(x_{t_{j}},Z_{j,n_j})
+
\sum_{i=1}^{n_j-1}
w_{j}(\dot{X}_{j,i},Z_{j,i})
}
\right],\;\text{for } q_J(x_{t_{J}}|x_{t_{J+1}}):=q_J(x_T),
\end{equation}
\[
w_0(x)=\pi_u(x), \quad
w_j(x,z)=\sum_{l=1}^{m_j}
\pi_u(z_{l})\,\frac{\tilde{q}_{t_j}(z_{l})}{q_{t_j}(z_{l})},
\;\; 1\leq j\leq J-1, \quad w_J(x,z)=\frac{1}{\bar{q}_T(x_T)}\sum_{l=1}^{m_
j}
\pi_u(z_{l})\,\frac{\tilde{q}_{T}(z_{l})}{q_{T}(z_{l})}
\]
where
$
\dot{X}_{j,i}\overset{i.i.d.}{\sim}
\mathcal{N}\!\left(
e^{\Delta t_j/2}X_{t_{j+1}},\,
(e^{\Delta t_j}-1)I_d
\right)$, $j=0,\ldots,J-1$; $\dot{X}_{J,i}\overset{i.i.d.}{\sim}\bar{q}_T$;
$Z_{j,i}=\{Z_{j,i,l}\}_{l=1}^{m_j}$,
$Z_{j,i,l}\overset{i.i.d.}{\sim}q_{t_j}$,
$j=1,\ldots,J$,
and $w_0(x,z):=w_0(x)$.

An unbiased estimator for $q_S$ is obtained by replacing the expectation in \eqref{eq.qjSP} with its Monte Carlo realization. Importantly, this estimator is almost surely positive and bounded, which enables the construction of exact MCMC updates targeting the true distribution.

Although this suggests pseudo-marginal Metropolis--Hastings, that route is not attractive here. It would require augmenting the state space with all auxiliary variables used to estimate the proposal density, and the intractable object is the proposal rather than the target, so that a standard pseudo-marginal construction would require an unbiased estimator of a reciprocal proposal density, therefore demanding significant extra auxiliary variables and computational effort. We instead use Barker's rule \citep{barker1965}, whose acceptance probability can be simulated by a Bernoulli factory using unbiased estimators of the proposal density itself.
For a target density $\pi$ and a proposal density $q$, both w.r.t. the same dominating measure, Barker's acceptance probability for a move $x \to y$ is
\begin{equation}\label{eq.:Bkrap}
\alpha_B(x,y)
=
\frac{p(y)q(x)}{p(y)q(x)+p(x)q(y)}.
\end{equation}

Although Metropolis-Hastings dominates Barker's in the Peskun ordering \citep{peskun1973}, the latter is at least half as efficient as MH in terms of asymptotic variance and inherits a CLT whenever MH does \citep{latuszynski2013}.

For the Ping-Pong algorithm, the target density w.r.t. $\bar Q$ is given by~\eqref{eq:RNrev_OU_times}. Therefore, the Barker acceptance probability of a SPIDER proposal $x\to\ddot x$ is
\[
\alpha_S(x,\ddot x)
=
\frac{
\{\pi_u(\ddot x_0)/\bar q_T(\ddot x_T)\}q_S(x)
}{
\{\pi_u(\ddot x_0)/\bar q_T(\ddot x_T)\}q_S(x)
+
\{\pi_u(x_0)/\bar q_T(x_T)\}q_S(\ddot x)
}.
\]
An $\alpha_S(x,\ddot x)$-coin can be simulated using the divide-and-conquer Bernoulli factory (DCBF) of \citet{DCBF2025}, which exploits the product structure of $q_S$ to improve computational efficiency. Details about DCBF for the Ping-Pong are presented in Appendix \ref{ssec.:DCBF}. The general algorithm is described below.

\begin{algorithm}[H]
\caption{Ping-Pong MCMC}\label{alg:pingpong}
\KwData{Schedule rule; particles; estimator sample sizes; Bernoulli factory specifications; number of iterations $M$.}
\KwResult{MCMC sample from the marginal law of $X_0$}

Initialize a path $x^{(0)}$, for instance using SPARK\;
\For{$k=1,\ldots,M$}{
    Select a schedule and one of the three admissible updates:
        
       sample $X_{t_1},\ldots,X_T\mid X_0=x_0^{(k-1)}$ exactly from~\eqref{eq:forward}\;

       propose $X_0,\ldots,X_{t_{J-1}}\mid X_T=x_T^{(k-1)}$ using SPIDER and accept by Barker's rule\;

        propose $\ddot{x} \mid x_T^{(k-1)}$ using SPIDER and accept by Barker's rule\;

    If the schedule changes, sample the required path values from the corresponding OU bridges\;
}
\Return{$\{x_0^{(k)}\}_{k=1}^M$}
\end{algorithm}

The resulting chain is ergodic under standard domination and positivity conditions, although geometric convergence need not hold because the relevant importance ratios may be unbounded. This is not a serious limitation in the intended use of the method: the chain is initialized from a distribution close to the target and acts primarily as an exact correction of a high-quality path proposal, rather than as a generic global exploration mechanism. Practical efficiency is therefore governed by the quality of the SPIDER proposal, and hence by the terminal time, time schedule, particle numbers and estimator sample sizes. These quantities should be calibrated to keep the importance weights stable and the Barker's acceptance probabilities sufficiently high, while controlling the computational cost of the Bernoulli factory. Further discussions about the Ping-Pong algorithms are provided in Appendix \ref{ssec.:MCMCanls}.

Finally, to use Ping-Pong within a larger Gibbs sampler, one sets $p_0$ equal to the full conditional density of the block being updated. Conditional on the remaining variables, the reverse-diffusion construction defines the corresponding path law on $[0,T]$, and the Ping-Pong kernel targets this extended law while preserving the desired full conditional of $X_0$ as its marginal invariant distribution. Thus the resulting update is an exact Barker's-within-Gibbs block update.

\section{Numerical Examples} \label{sec.experiments}

We present two examples to illustrate and empirically explore the
properties of the proposed algorithms.

\paragraph{Bidimensional radial distribution}
Define a distribution on $\mathbb{R}^2$ by sampling a point in polar
coordinates $(\theta, z)$, where $\theta \sim \mathrm{U}[0, 2\pi]$ and
$z$ follows a mixture of normals with means $(3, 6, 9, 12)$, variance
$0.1^2$, and weights $(0.1, 0.4, 0.1, 0.4)$.  The resulting density is
\[
  p_0(x) \propto \frac{1}{\|x\|}
  \sum_{j=1}^{4} w_j \exp\!\left(-\frac{(\|x\| - 3j)^2}{2\cdot 0.1^2}\right).
\]

Both SPARK and Ping-Pong are run with the time schedule
$(0.1,\, 0.04,\, 0.02,\, 0.01,\, 0.005,\, 0.002,\, 0.0006,\, 0)$
and the standardized version of $p_0$.  For SPARK, only the
reverse-time transition block is used, with an estimator sample size of
$50$.  All pCN proposals use $\beta = 0.99$ and each chain runs for
$100$ iterations.  The SPIDER proposal in Ping-Pong uses $5$ particles
and an estimator sample size of $10$.  SPARK generates $5{,}000$
samples; Ping-Pong runs $15{,}000$ iterations and achieves an acceptance
rate of approximately $0.3$.

Both algorithms are run on an M2 MacBook Air using $7$ cores in
parallel.  SPARK completes $5{,}000$ independent samples in $15$
seconds; Ping-Pong completes $15{,}000$ iterations in $1.6$ minutes.

\begin{figure}[ht!]
\centering
\includegraphics[width=0.98\linewidth]{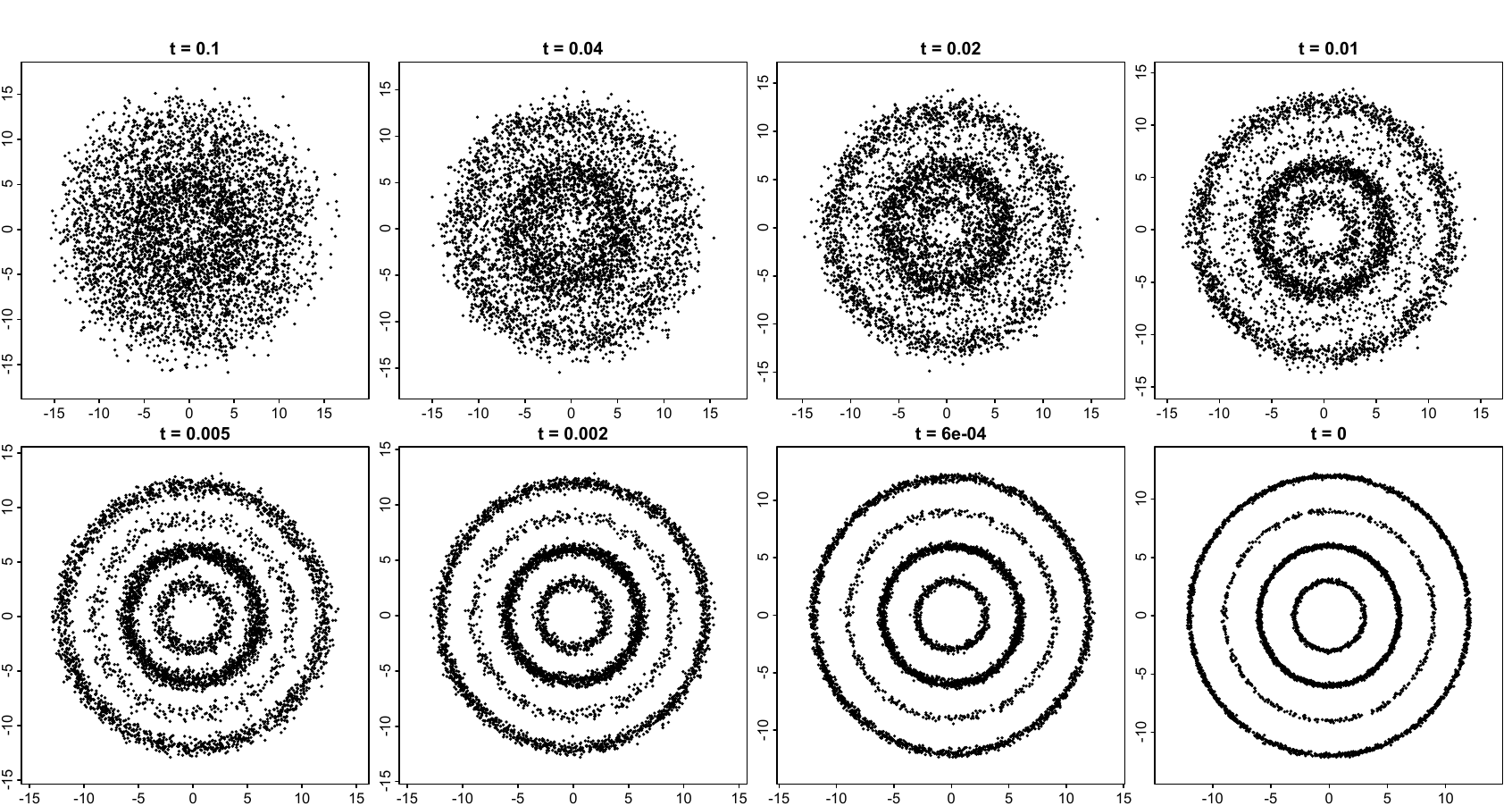}
\caption{Samples generated by SPARK at each time in the schedule.}
\label{fig:spark_radial}
\end{figure}

\begin{figure}[ht!]
\centering
\includegraphics[width=0.98\linewidth]{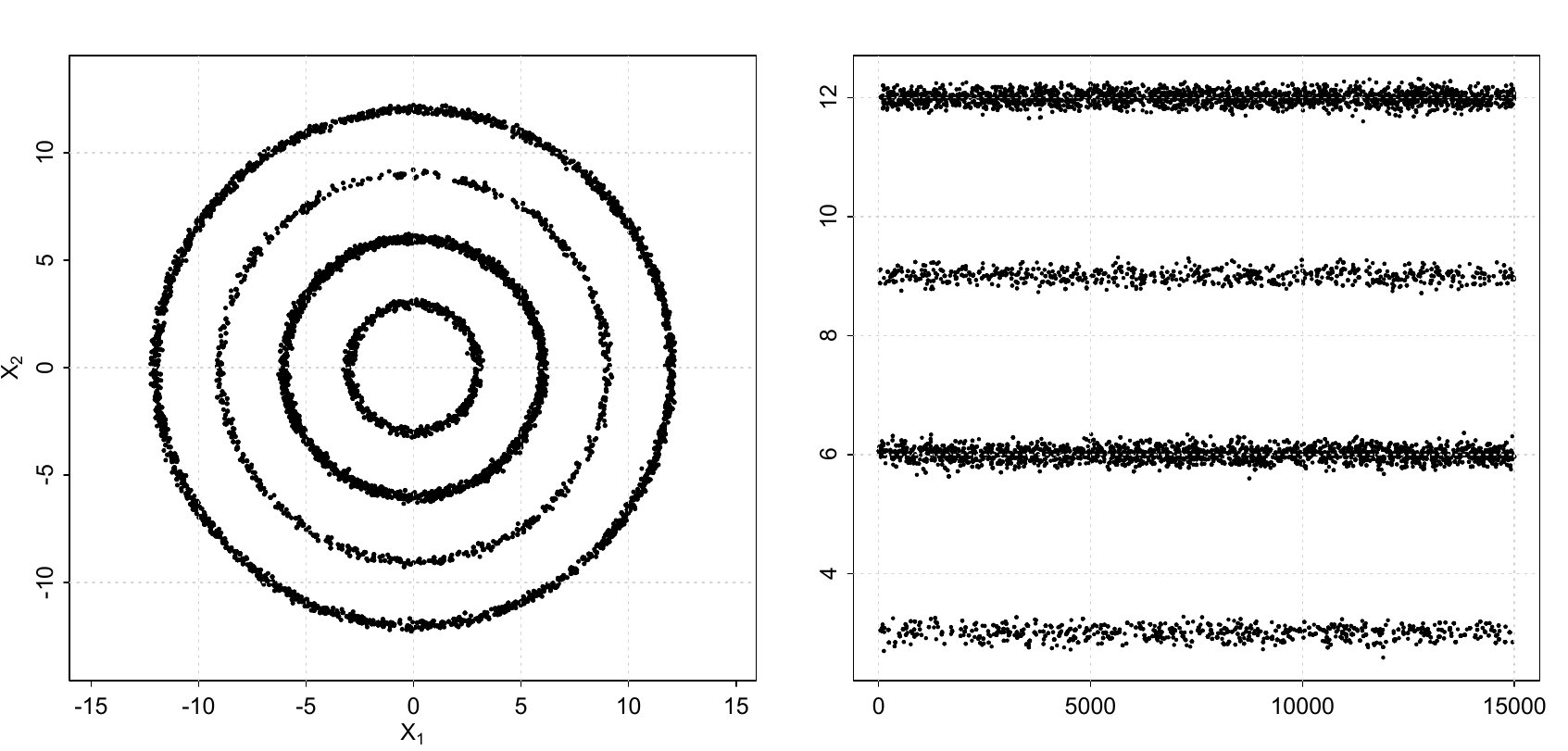}
\caption{Sample (left) and trace plot of the radius (right) generated by Ping-Pong.}
\label{fig:pp_radial}
\end{figure}

\paragraph{Ten-dimensional mixture of normals}
Consider an equal-weight mixture of three ten-dimensional normal
distributions with means $\mu_k = m_k \mathbf{1}_{10}$, where
$(m_1, m_2, m_3) = (-\sqrt{10},\, 0,\, \sqrt{10})$, and
covariance $0.3\, I_{10}$.

SPARK is implemented with the time schedule
$(1.7,\, 1.5,\, \ldots,\, 0.3,\, 0.15,\, 0.1,\, 0.07,\, 0.04,\, 0.02,
 \, 0.015,\, 0.01,\, 0)$
and the standardized version of $p_0$.  Both MC blocks --reverse-time
transition and marginals-- are used, with estimator sample sizes ranging
from $2{,}500$ at time $1.7$ to $250$ at time $0.01$.  Each PMMH chain
runs for $100$ iterations, and a final sample of size $5{,}000$ is
produced.  Total running time is $8.6$ minutes.

\begin{figure}[ht!]
\centering
\includegraphics[width=0.5\linewidth]{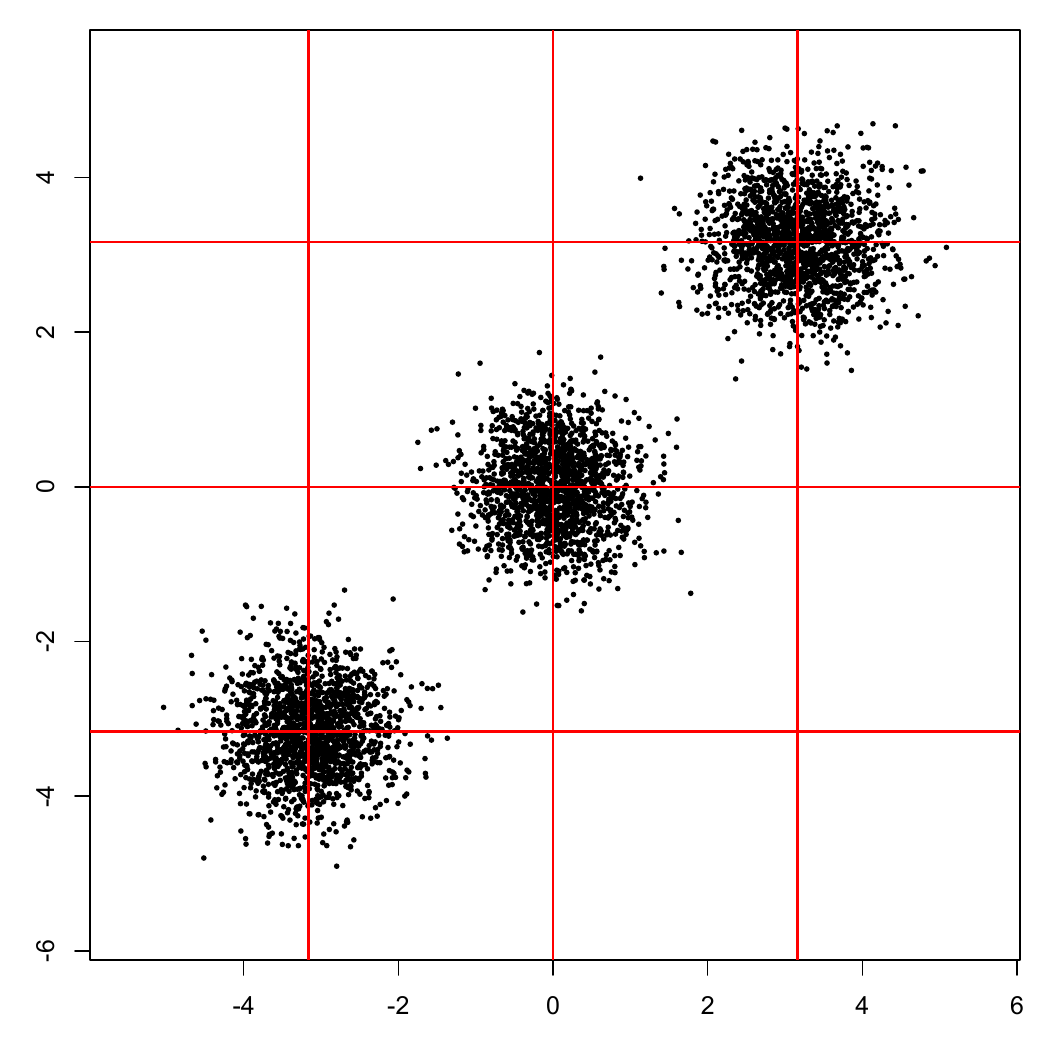}
\caption{Sample generated by SPARK projected onto the first two
  coordinates.  Empirical component weights are $0.340$, $0.326$, and
  $0.334$.  Red crosshairs indicate the true component means.}
\label{fig:spark_mixture}
\end{figure}

\section{Conclusion}

This paper introduces a new perspective on the use of reverse-time diffusions for sampling from unnormalized densities. The central novelty of the approach lies in two key features. First, the proposed methods avoid any time-discretization error, as they rely on exact representations of the underlying diffusion dynamics. Second, they are \emph{zeroth-order} methods, in the sense that they do not require the evaluation or estimation of derivatives of the log-density. This distinguishes them from a large class of existing diffusion-based samplers whose implementation depends on gradient information.

The methodology developed in this work builds on the theory of exact simulation for diffusion processes. In particular, the algorithms rely on a change-of-measure relationship between the transition law of a reverse-time diffusion and that of a Gaussian Markov process. A key ingredient of the approach is the availability of unbiased estimators for the marginal density of the diffusion evaluated at a collection of intermediate times. These estimators allow us to construct valid acceptance probabilities and importance weights without introducing bias, and play a central role in both classes of algorithms proposed in the paper.

Two main algorithmic frameworks are developed. The first is based on a sequence of MCMC algorithms with the last one converging to the original target under consideration. The second consists of a single MCMC chain targeting an augmented space that has the target as its marginal. While the latter is more precise in the sense of relying on the convergence of only one Markov chain, the former has broader application.

As is common with methods based on stochastic estimators and exact simulation techniques, the algorithms involve several tuning parameters whose specification can strongly influence performance. This paper provides practical guidelines for selecting these quantities, supported by theoretical analysis that clarifies the role of the main parameters controlling the variance of the estimators and the stability of the algorithms. These results help illuminate the trade-offs between computational cost and statistical efficiency.

The proposed methods also present some limitations that should be acknowledged. In particular, the algorithms may incur a non-negligible computational cost due to repeated evaluations of the target density and the use of unbiased estimation procedures, including Bernoulli factory constructions. These operations can be computationally intensive in certain settings. Nevertheless, it is important to emphasize that the resulting algorithms remain zeroth-order procedures that do not require gradients of the log-density. In many applications where derivative information is unavailable, unreliable, or expensive to compute, this property can represent a substantial practical advantage compared to state-of-the-art gradient-based diffusion samplers.

The empirical results presented for simulated examples illustrate the practical behavior of the proposed algorithms, highlighting their robustness and providing insight into the impact of the main tuning parameters.

More broadly, this work opens a new avenue in the development of reverse-time diffusion-based sampling methods. By combining ideas from exact diffusion simulation, unbiased estimation, and extended-space MCMC methodology, the framework proposed here provides a flexible foundation for future research. We expect that the techniques introduced in this paper may inspire further developments, including improved estimators, more efficient algorithmic implementations, and extensions to broader classes of stochastic processes and target distributions.

\bibliographystyle{apalike} 
\bibliography{biblio} 

\appendix

\renewcommand{\thetheorem}{\Alph{section}.\arabic{theorem}}
\setcounter{theorem}{0} 

\renewcommand{\theproposition}{\Alph{section}.\arabic{proposition}}
\setcounter{proposition}{0}

\section{Proofs} \label{ap:proofs}

\begin{proof}[Proof of Proposition \ref{prop_TV}]
Part $(i)$: the first inequality follows from Pinsker's inequality and the second one from Theorem 5.2.1 in \cite{bakry2014analysis} and the Gross logarithmic Sobolev inequality.

Part $(ii)$:
Notice that he process $(X_t)_{t \geq 0}$ has a closed-form solution,
\[ X_t = e^{-t/2}X_0 + e^{-t/2}\int_0^t e^{s/2}\,dB_s. \]
Then, $X_t\mid X_0=x_0 \sim \mathcal N(e^{-t/2}x_0,(1-e^{-t})I_d) $, and the result follows straightforwardly.

\end{proof}

\begin{proof}[Proof of Theorem \ref{th:girsanovgeral}]
Let $q(s,x)=\log p(T-s,x)$, $\bar f(s,x)=f(T-s,x)$, and $\bar g(s)=g(T-s)$.  
By the Fokker-Planck equation,
\[
\partial_s q
=
\nabla\!\cdot\bar f
+
\bar f\cdot\nabla q
-
\frac{\bar g^2}{2}\big(\Delta q+\|\nabla q\|^2\big)=0.
\]

Applying It\^o's formula to $q(s,Y_s)$ gives
\begin{align*}
q(s_1,Y_{s_1})-q(s_0,Y_{s_0})
&=
\int_{s_0}^{s_1}
\Big[
\partial_s q(s,Y_s)
-
\nabla q(s,Y_s)\!\cdot\!\bar f(s,Y_s)
+
\frac{\bar g(s)^2}{2}\Delta q(s,Y_s)
\Big]ds
+
\int_{s_0}^{s_1}\bar g(s)\nabla q(s,Y_s)\,dB_s .
\end{align*}

Substituting the Fokker-Planck identity yields
\begin{align*}
q(s_1,Y_{s_1})-q(s_0,Y_{s_0})
&=
\int_{s_0}^{s_1}\nabla\!\cdot\bar f(s,Y_s)\,ds
-
\frac12\int_{s_0}^{s_1}\bar g(s)^2\|\nabla q(s,Y_s)\|^2ds
+
\int_{s_0}^{s_1}\bar g(s)\nabla q(s,Y_s)\,dB_s .
\end{align*}

Exponentiating both sides gives
\[
\exp\!\left(
\int_{s_0}^{s_1}\bar g(s)\nabla q(s,Y_s)\,dB_s
-
\frac12\int_{s_0}^{s_1}\bar g(s)^2\|\nabla q(s,Y_s)\|^2ds
\right)
=
e^{-\int_{s_0}^{s_1}\nabla\!\cdot f(T-s,Y_s)ds}
\frac{p(T-s_1,Y_{s_1})}{p(T-s_0,\bar{x}_0)}.
\]

The result then follows from Girsanov's theorem.
\end{proof}

\begin{proof}[Proof of Corollary \ref{Corl1}]

Substituting $f(s,x)=-\tfrac{1}{2}x$ and $g(s)=1$ into the Radon-Nikodym expression of Theorem~\ref{th:girsanovgeral} yields the stated Gaussian transition and the identity
\[
\frac{d\bar P}{d\bar Q}(\omega)
= \exp\!\left\{\tfrac{d}{2}(s_1-s_0)\right\}
\frac{p(T-s_1,\omega_{s_1})}{p(T-s_0,\omega_{s_0})}.
\]

To establish parts (i) and (ii), consider the conditional measures 
$\tilde{\bar P}(\cdot)=\bar P(\cdot\mid\omega_{s_1})$ and 
$\tilde{\bar Q}(\cdot)=\bar Q(\cdot\mid\omega_{s_1})$.  
By the chain rule for Radon-Nikodym derivatives,
$\frac{d\bar P}{d\bar Q}(\omega) = \frac{d\bar P}{d\bar Q}(\omega_{s_1})
\frac{d\tilde{\bar P}}{d\tilde{\bar Q}}(\omega)$.

Since the expression for $\frac{d\bar P}{d\bar Q}(\omega)$ depends on $\omega$ only through $\omega_{s_1}$, it follows that $\ds \frac{d\tilde{\bar P}}{d\tilde{\bar Q}}(\omega)=1$,
which proves (i) and (ii).  

\end{proof}

\begin{proof}[Proof of Proposition \ref{prop:targt_mcmc}]

Define $\omega_{A}=\{\omega_t,\;t\in A\}$, we have
\[
\frac{d\bar P}{d\bar Q}(\omega_{[0:J]}) = \frac{p(T,\omega_0)}{q(\omega_0)}
\frac{d\bar P}{d\bar Q}(\omega_{(0:J)}) \propto \frac{p(T,\omega_0)}{q(\omega_0)}
\frac{p_0(\omega_T)}{p(T,\omega_0)} = \frac{p_0(\omega_T)}{q(\omega_0)}.
\]
\end{proof}

\begin{proof}[Proof of Proposition \ref{targino}]
Since $X_t \mid X_0=x_0 \sim \mathcal N(e^{-t/2}x_0,(1-e^{-t})I_d)$,
we have that
\begin{align*}
p(t,x)
&= \int_{\mathbb R^d} p_0(z)\,(2\pi(1-e^{-t}))^{-d/2}
\exp\!\Big(-\frac{\|x-e^{-t/2}z\|^2}{2(1-e^{-t})}\Big)\,dz = c(t)\,\mathbb E_{Z_{t,x}}\!\left[p_0(Z_{t,x})\right].
\end{align*}
\end{proof}

\begin{proof}[Proof of Proposition \ref{prop.GTEst}]
Let $v=e^t-1$ and recall
$
\ds \Var(g)=\int \pi(y)^2 \frac{p_x(y)^2}{q_t(y)}\,dy - r^2
$.
It suffices to show finiteness of the integral.

Since $\pi$ is bounded on $\{\|y\|\le R\}$, the integral over this region is finite. For $\|y\|>R$,
$\ds
\pi(y)^2 \le C^2 e^{-2a\|y\|^2}
$.
Moreover, for Gaussian densities,
$\ds
\frac{p_x(y)^2}{q_t(y)} \asymp \exp\!\left\{\left(-\frac{1}{v}+\frac{1}{2\tau^2}\right)\|y\|^2 + O(\|y\|)\right\}
$,
where $\tau^2=\frac{\sigma^2 v}{\sigma^2+v}$.

Hence, for large $\|y\|$,
\[
\pi(y)^2 \frac{p_x(y)^2}{q_t(y)}
\lesssim
\exp\!\left\{\left(-2a-\frac{1}{v}+\frac{1}{2\tau^2}\right)\|y\|^2 + O(\|y\|)\right\}.
\]
Thus integrability holds if
$\ds
-2a-\frac{1}{v}+\frac{1}{2\tau^2}<0
\;\Longleftrightarrow\;
\frac{1}{2\tau^2}<2a+\frac{1}{v}
$.
Substituting $\tau^2=\frac{\sigma^2 v}{\sigma^2+v}$ gives
$\ds 
\frac{\sigma^2+v}{2\sigma^2 v}<2a+\frac{1}{v}
\;\Longleftrightarrow\;
\sigma^2>\frac{v}{1+4av}
$.
This ensures finiteness of the second moment, hence $\Var(g)<\infty$.

\end{proof}

\begin{proof}[Proof of Proposition \ref{prop.stTEst}]
Let
$\ds
g(Y)=\frac{\pi(Y)\tilde q_t(Y)}{q_t(Y)}$, $Y\sim q_t$.
Then $\E[g(Y)]=r$ and
$\ds
\E[g(Y)^2]
=
\int \frac{\pi(y)^2 \tilde q_t(y)^2}{q_t(y)}\,dy
$.

Using $q^\star(y)=\frac{\pi(y)\tilde q_t(y)}{r}$,
$\ds
\E[g(Y)^2]
=
r^2 \int \frac{q^\star(y)^2}{q_t(y)}\,dy
$.
Hence
$\ds
\Var(g(Y))
=
r^2\left(\int \frac{q^\star(y)^2}{q_t(y)}\,dy - 1\right)
$.

Finally,
$\ds
\chi^2(q^\star\|q_t)
=
\E_{q_t}\!\left[\left(\frac{q^\star}{q_t}-1\right)^2\right]
=
\int \frac{q^\star(y)^2}{q_t(y)}\,dy -1
$,
which gives
$\ds
\Var(g(Y))=r^2\,\chi^2(q^\star\|q_t)
$.

\end{proof}


\section{Bernoulli Factory for the Ping-Pong Algorithms}\label{ssec.:DCBF}

The Metropolis-type updates within the \emph{Ping-Pong} class require the simulation of Bernoulli random variables with acceptance probabilities $\alpha_S$ or $\alpha_{T,S}$, depending on the update used. 
These probabilities are not available in closed form, but can be expressed in terms of quantities for which unbiased estimators can be simulated. 
Consequently, acceptance is implemented through a \emph{Bernoulli factory}.

\paragraph{Bernoulli factory formulation.}

Let $p=(p_0,p_1)$ denote unknown probabilities for which we can simulate independent $p_i$-coins, $i=0,1$. 
A Bernoulli factory problem consists of simulating an $f(p)$-coin, that is, a Bernoulli random variable with success probability $f(p)\in(0,1)$, using only $p_i$-coins.

In our setting, the Barker acceptance probability has the form
\begin{equation}\label{eq.:Bkrap2}
f(p)=\frac{c_1 p_1}{c_0 p_0 + c_1 p_1},
\end{equation}
where $c_0,c_1>0$ are known constants and $p_0,p_1$ are unknown probabilities.
This structure admits the so-called \emph{2-coin algorithm}, originally proposed in \cite{gonccalves2023exact} and later refined in \cite{vats2022efficient,DCBF2025}.

\paragraph{Simulation of $p_i$-coins.}
In our setting, we do not observe $p_i$ directly, but can simulate unbiased estimators $\hat p_i$ satisfying $\mathbb{E}[\hat p_i]=p_i$ and $\hat p_i\in[0,1]$ almost surely. 
A $p_i$-coin can then be simulated exactly by:
\begin{enumerate}[nosep]
\item Simulating $\hat p_i$;
\item Given $\hat p_i$, drawing a Bernoulli$(\hat p_i)$ variable.
\end{enumerate}
By iterated expectation, the resulting coin has success probability $p_i$.

\paragraph{The 2-coin algorithm.}
The 2-coin algorithm exploits representation \eqref{eq.:Bkrap2}. 
At each iteration, one first selects index $i\in\{0,1\}$ with probability $c_i/(c_0+c_1)$, then simulates a $p_i$-coin.
If the $p_i$-coin is successful, the algorithm outputs $i$; otherwise it restarts.
Algorithm~\ref{alg:2coin} presents the procedure for an $f(p)$-coin.

\begin{algorithm}[H]
\caption{2-coin algorithm}\label{alg:2coin}
\KwData{Constants $c_0,c_1>0$; mechanism to simulate $p_i$-coins}
\KwResult{An $f(p)$-coin}
\While{no output}{
Draw $I \sim \text{Ber}\!\left\{c_1/(c_0+c_1)\right\}$\;
Simulate a $p_I$-coin\;
\eIf{$p_I$-coin $=1$}{
Output $I$\;
}{
continue\;
}}
\end{algorithm}

\cite{vats2022efficient} propose a generalization of Barker’s 2-coin algorithm in which a $\beta$-coin is simulated at the beginning of each iteration of the 2-coin procedure. 
If this $\beta$-coin returns 1, the algorithm immediately outputs 0; otherwise, the usual 2-coin step is performed. 
This modification effectively truncates the number of loops of the 2-coin algorithm, thereby preventing occasionally very costly realizations that may arise when the underlying coin probabilities are small.

The choice of $\beta$ reflects a trade-off between computational cost and mixing efficiency. 
Larger values of $\beta$ reduce the expected cost per invocation of the 2-coin algorithm by increasing the probability of early termination. 
However, this comes at the price of additional rejections, which slow down the mixing of the Markov chain. 
Thus, $\beta$ should be calibrated to balance computational savings against the induced degradation in mixing performance.

\paragraph{Product structure and divide-and-conquer.}

In both Barker's updates (Update~2 and Update~3), the acceptance probabilities admit a product representation of the form
\[
c_i p_i
=
\prod_{j=1}^{J^\star}
 c_i^{1/J^\star} p_{ij},
\]
where $J^\star = J$ for Update~2 and $J^\star = J+1$ for Update~3.
A $p_i$-coin can therefore be simulated by generating independent $p_{ij}$-coins and returning $1$ if and only if all sub-coins equal $1$.

Directly applying the 2-coin algorithm to simulate an $f(p)$-coin when 
$p_i=\prod_{j=1}^J p_{ij}$ leads to an exponential cost in $J$ 
(see \cite{DCBF2025}).
The divide-and-conquer Bernoulli factory (DCBF) of \cite{DCBF2025} 
reduces this cost by organizing the $J$ factors along a balanced binary tree. 
The index set $\{1,\ldots,J\}$ is partitioned across the leaves of the tree, 
and at each leaf $l$ the 2-coin algorithm is used to simulate an 
$f_l(p)$-coin depending only on the corresponding subset of factors. 
More precisely, if $l \subset \{1,\ldots,J\}$ denotes the subset of indices 
assigned to leaf $l$, then
\[
f_l(p)
=
\frac{\prod_{j\in l} c_{1}^{1/J} p_{1j}}
 {\prod_{j\in l} c_{1}^{1/J} p_{1j}
  + 
  \prod_{j\in l} c_{0}^{1/J} p_{0j}}.
\]

After simulating the $f_l(p)$-coin at each leaf, the outputs are propagated 
upward through the tree. At each internal node, the two children are 
compared: if they agree (both $0$ or both $1$), their common value is 
assigned to the parent node; if they differ, the parent node remains 
unresolved. 

Once this upward pass is completed, some internal nodes may remain 
unresolved due to disagreement between their children. The leaves 
connected to such unresolved nodes identify precisely the subsets of 
indices that must be re-simulated, while all other leaves retain their 
previously simulated values.

If the root node becomes resolved, its value is returned as the 
$f_l(p)$-coin. Otherwise, only the leaves associated with unresolved 
branches are re-simulated, and the procedure repeats. 
This recursive propagation and localized re-simulation prevents 
repeated evaluation of all $J$ factors at each restart and is the key 
mechanism reducing the overall computational complexity from exponential 
to polynomial in $J$ under the regimes studied in \cite{DCBF2025}.
Although developed in a Bayesian asymptotic setting, the same structure substantially reduces computational cost in the present context.

Algorithm~\ref{alg:DCBF} outlines the DCBF construction.

\begin{algorithm}[H]
\caption{Divide-and-Conquer Bernoulli Factory (DCBF)}\label{alg:DCBF}
\KwData{Constants $c_i$; number of levels $\ell$; mechanism to simulate $f_l(p)$-coins via the 2-coin algorithm}
\KwResult{An $f(p)$-coin}
Construct a balanced binary tree with $2^{\ell}$ leaves\;
Partition $\{1,\ldots,J\}$ across the leaves as evenly as possible, defining the subsets $l$\;

Mark all leaves as active\;

\While{the root node is unresolved}{
Simulate an $f_l(p)$-coin at each active leaf $l$\;

Propagate the leaf outputs upward through the tree as described in the text\;

\eIf{the root node becomes resolved}{
Output its value as the $f(p)$-coin\;
}{
Mark as active only the leaves belonging to unresolved branches\;
}
}
\end{algorithm}

\paragraph{DCBF for Update~2.}

For Update~2, we have $j=0,\ldots,J-1$ and
\begin{equation}\label{coins_prob1}
c_{0j}=(\pi_u(x_0))^{1/J}, 
\qquad
\hat p_{0j}
=
\frac{w_j(x_{t_j},Z_{j,n_j})}
 {w_j(x_{t_j},Z_{j,n_j})+\sum_{i=1}^{n_j-1}w_j(\dot X_{j,i},Z_{j,i})}.
\end{equation}
The expressions for $c_{1j}$ and $\hat p_{1j}$ are obtained by replacing $x$ with $\ddot x$.

\paragraph{DCBF for Update~3.}

For Update~3, we set $j=0,\ldots,J$ and
\begin{equation}\label{coins_prob2}
c_{0j}=\left(\frac{\pi_u(x_0)}{\tilde q_T(x_T)}\right)^{1/(J+1)}.
\end{equation}
For $j=0,\ldots,J-1$, the estimators $\hat p_{0j}$ are defined exactly as in Update~2.

\section{Analysis of the Ping-Pong algorithms}\label{ssec.:MCMCanls}

This section analyzes the theoretical properties of the \emph{Ping-Pong} class of MCMC algorithms and discusses how these properties guide practical implementation choices. Some factors influencing efficiency are shared with SPARK, such as the terminal time $T$, the time schedule and the Monte Carlo sample size used in the density estimators. Others are specific to the MCMC context, including the number of particles in the SPIDER proposal, the choice of update combination within the Ping-Pong class and the specifications of the DCBF algorithm.

\subsubsection*{Proposal distribution}

The densities of the proposal distributions within the \emph{Ping-Pong} class are intractable because the distribution of the SPIDER output is itself intractable. Nevertheless, unbiased estimators can be constructed, based on the following result.

\begin{proposition}\label{prop.ueSP_IS}
Let $\tilde Q$ denote the SPIDER law associated with the IS estimator
\eqref{eq:unb_est2}, with $q_t=q_t^G$, completed with bridges from $\bar Q$, and assume that $X_T$ has Lebesgue density $\bar q_T$ under $\bar Q$. Then
\[
q_S(x):=\frac{d\tilde Q}{d\bar Q}(x)
=
q_J(x_T)\prod_{j=J-1}^0 q_j(x_{t_j}\mid x_{t_{j+1}}),
\]
where, for $j=0,\ldots,J-1$,
\begin{equation}\label{eq.qjSP_IS}
q_j(x_{t_j}\mid x_{t_{j+1}})
=
n_j\,
\mathbb E_{\dot X_j,Z_j}
\left[
\frac{
w_j(x_{t_j},Z_{j,n_j})
}{
w_j(x_{t_j},Z_{j,n_j})
+
\sum_{i=1}^{n_j-1} w_j(\dot X_{j,i},Z_{j,i})
}
\right],
\end{equation}
with
\[
w_0(x)=\pi_u(x),
\quad
w_j(x,z)=\sum_{l=1}^{m_j}
\pi_u(z_l)\frac{\tilde q_{t_j}(z_l)}{q_{t_j}(z_l)},
\;\;1\le j\le J-1,\;\; 
w_J(x,z)=\frac{1}{\bar q_T(x_T)}\sum_{l=1}^{m_J}
\pi_u(z_l)\frac{\tilde q_T(z_l)}{q_T(z_l)},
\]
and
\[
\dot X_{j,i}\overset{i.i.d.}{\sim}
\mathcal N\!\left(
e^{\Delta t_j/2}X_{t_{j+1}},\,
(e^{\Delta t_j}-1)I_d
\right),
\quad
\dot X_{J,i}\overset{i.i.d.}{\sim}\bar q_T,
\quad
Z_{j,i,l}\overset{i.i.d.}{\sim}q_{t_j},
\qquad
j=1,\ldots,J.
\]
\end{proposition}

\begin{proof}[Proof of Proposition \ref{prop.ueSP_IS}]
Under SPIDER, the path is generated sequentially backward. It suffices to
identify, for each $j$, the conditional density of $X_{t_j}$ given
$X_{t_{j+1}}$ with respect to the proposal kernel
\[
\bar Q_j(x_{t_{j+1}},dx_{t_j})
=
\mathcal N\!\left(
e^{\Delta t_j/2}x_{t_{j+1}},
(e^{\Delta t_j}-1)I_d
\right)(dx_{t_j}),
\quad
\text{for } j<J, \text{ and } \bar q_T \text{ for } j=J.
\]

Conditionally on $X_{t_{j+1}}=x_{t_{j+1}}$, the algorithm draws
$\ds
\dot X_{j,1},\ldots,\dot X_{j,n_j}
\overset{i.i.d.}{\sim}
\bar Q_j(x_{t_{j+1}},\cdot)$,
together with independent auxiliary variables $Z_{j,1},\ldots,Z_{j,n_j}$ with
$\ds
Z_{j,i,l}\overset{i.i.d.}{\sim}q_{t_j},
$
and selects an index $I_j$ with probability
\[
\mathbb P(I_j=k\mid \dot X_j,Z_j)
=
\frac{w_j(\dot X_{j,k},Z_{j,k})}
{\sum_{i=1}^{n_j}w_j(\dot X_{j,i},Z_{j,i})},
\]
where $w_j$ includes the IS correction
$\tilde q_{t_j}/q_{t_j}$.
Thus $X_{t_j}=\dot X_{j,I_j}$. For any measurable $A$,
\begin{align*}
\tilde{Q}(X_{t_j}\in A\mid X_{t_{j+1}})
&=
\sum_{k=1}^{n_j}
\mathbb E\!\left[
\mathbf 1_{\{\dot X_{j,k}\in A\}}
\frac{w_j(\dot X_{j,k},Z_{j,k})}
{\sum_{i=1}^{n_j}w_j(\dot X_{j,i},Z_{j,i})}
\right]\\
&=
n_j\,
\mathbb E\!\left[
\mathbf 1_{\{\dot X_{j,n_j}\in A\}}
\frac{w_j(\dot X_{j,n_j},Z_{j,n_j})}
{w_j(\dot X_{j,n_j},Z_{j,n_j})+\sum_{i=1}^{n_j-1}w_j(\dot X_{j,i},Z_{j,i})}
\right]\\
&=
n_j\,
\mathbb E\!\left[
\int_A
\frac{w_j(x,Z_{j,n_j})}
{w_j(x,Z_{j,n_j})+\sum_{i=1}^{n_j-1}w_j(\dot X_{j,i},Z_{j,i})}
\,\bar Q_j(x_{t_{j+1}},dx)
\right]\\
&=
\int_A
n_j\,
\mathbb E_{\dot X_j,Z_j}\!\left[
\frac{w_j(x,Z_{j,n_j})}
{w_j(x,Z_{j,n_j})+\sum_{i=1}^{n_j-1}w_j(\dot X_{j,i},Z_{j,i})}
\right]
\bar Q_j(x_{t_{j+1}},dx).
\end{align*}

This identifies $q_j(\cdot\mid x_{t_{j+1}})$ as the Radon--Nikodym derivative
in \eqref{eq.qjSP_IS}.
\end{proof}

\subsubsection*{Convergence and ergodicity}

For independence Metropolis-Hastings algorithms with proposal density $q$
and invariant distribution $\pi$, the associated Markov kernel $P$ satisfies:

\begin{enumerate}
    \item
    \[
    \lim_{n\to\infty}
    \left\|
    P^n(x,\cdot) - \pi
    \right\|_{\mathrm{TV}}
    = 0,
    \qquad \text{for $\pi$-a.e. } x\in\mathcal{X},
    \]
    
    \item
    \[
    \frac{1}{n}\sum_{i=1}^n h(X_i)
    \;\longrightarrow\;
    \pi(h),
    \qquad \text{$P$-a.s.,}
    \]
\end{enumerate}
for any measurable $h$ with $\pi(|h|)<\infty$,  provided \citep{DJKW2025}:

\begin{enumerate}[label=(\roman*)]
    \item $\pi \ll q$;
    \item $w=\pi/q$ satisfies $w>0$ $q$-a.s.\ and $q(w)=1$.
\end{enumerate}

These conditions apply directly to the pure independence formulation (Update 3), and analogous arguments extend to the two-block Gibbs-type formulation (Updates 1 and 2), since each Metropolis-type update satisfies the same domination and positivity conditions.

Furthermore, \cite{MengTwee1996} show that
geometric ergodicity holds if and only if the importance weight $w$ is essentially bounded.
Bounded weights imply a global minorization condition and hence uniform (and therefore geometric) ergodicity.
When $w$ is unbounded, geometric convergence in total variation fails.
An analogous statement holds for the independence Barker's kernel.

In our setting, $\pi_u$ is bounded, so the auxiliary weights $w_j(\cdot,\cdot)$ are uniformly bounded above.
This does not imply that the global importance ratios arising in Update~2 or Update~3 are bounded.
A single small factor in $q_S$ can make the product arbitrarily small.
Accordingly, we do not assume bounded importance ratios and therefore do not rely on uniform or geometric ergodicity.

Nevertheless, this is not problematic.
Uniform ergodicity requires worst-case control over all starting points.
In practice, the Ping-Pong algorithms are initialized from a distribution $\nu$ designed to approximate $\pi$.
Hence it is more relevant to study
\[
\|\nu P^t - \pi\|_{\mathrm{TV}}
\]
rather than worst-case bounds.

Moreover, the SPIDER proposal approximates the full diffusion path law.
The MCMC acts primarily as a correctness mechanism ensuring exactness, rather than as a global explorer of the state space.
In this regime, the acceptance rate becomes a structurally meaningful diagnostic:
sustained high acceptance reflects that the proposal already lies close to stationarity and that the MCMC performs only small perturbative corrections.
This interpretation differs from generic independence samplers, where high acceptance may arise from accidental tail matching.

A natural initialization $\nu$ is given by a SPARK realization.
This places the chain in a high-probability region of the target, aligning theoretical considerations with practical operation.

\subsubsection*{Terminal time $T$ and proposal $\boldsymbol{\tilde{q}}_T$}

The choice of $T$ plays a distinct role under the two canonical strategies.

Under the Updates~1 and~2 Gibbs sampling formulation, mixing is affected by inter-block dependence, specifically by the correlation between $X_0$ and $X_T$ under the target. As $T$ decreases, this dependence strengthens, reducing Gibbs efficiency.

Under Update~3 alone formulation, this block-correlation effect is absent. The choice of $T$ then follows the same rationale as in SPARK: $T$ must be large enough that the SIR proposal for $X_T$ produces reasonable candidates and hence satisfactory acceptance rates.

Because the second formulation avoids block correlation, it may permit smaller values of $T$, potentially improving overall efficiency.

Regarding the choice of the particle distribution $\bar{q}_T$ in the SIR step for $X_T$ (Update 3), a preliminary learning phase is generally beneficial. One practical approach is to run a PMMH algorithm targeting $X_T$ (as in SPARK) in order to obtain empirical information about salient features of $p(T,\cdot)$, such as location, scale, and potential residual dependence or tail behavior. Provided that $T$ is large enough to have mitigated strong multimodality and dependence, a multivariate Student-$t$ distribution with independent coordinates typically yields a robust and efficient SIR proposal. The heavier tails offer additional protection against moderate misspecification.

Under this perspective, the terminal time $T$ for Update~3 can be selected following the same calibration strategy proposed for SPARK. Specifically, $T$ should be large enough so that the marginal law $p(T,\cdot)$ exhibits sufficiently regular structure, with attenuated multimodality and weakened dependence, allowing a simple parametric proposal such as the Student-$t$ to approximate it effectively. At the same time, $T$ should not be taken unnecessarily large, since this increases computational cost without substantial gains in proposal quality. Thus, the choice of $T$ reflects a balance between structural simplification of the target at time $T$ and overall computational efficiency.

\subsubsection*{Time schedule and number of particles}

Given a fixed terminal time $T$, the intermediate time schedule should be chosen with care.
Although the MCMC correction guarantees exactness, and therefore compensates for the approximation error introduced by the self-normalized sequential SIR mechanism, the quality of the time discretization remains crucial for efficiency. 
The schedule determines how smoothly the reverse-diffusion bridge connects the target at time $0$ to the reference distribution at time $T$. 
If the intermediate times are poorly spaced, successive importance weights become more variable, degrading the quality of the SPIDER proposal and ultimately lowering acceptance probabilities in the MCMC correction step.

For this reason, the time schedule should be chosen so as to control the incremental inefficiency at each level.
In particular, adjacent time points should be sufficiently close to prevent large variability in the incremental importance ratios, thereby maintaining stable particle weights and a well-aligned global proposal.

An efficient way to measure the efficiency of each SIR algorithm and make it uniformly reasonable across all the time points is to evaluate the effective sample size of the particle approximation, commonly defined as $\ds \left(\sum_{k=1}^nw_{i}^2\right)^{-1}$, where $w_i$ are the exact self-normalized weights of the generate sample. In our case, these can be estimated using the respective estimator of the IS weights.

The additional robustness provided by the MCMC correction is more effectively exploited to reduce the number of particles, rather than to relax the time schedule. 
This is because the computational cost of the adopted Bernoulli factory is highly sensitive to the particle specification. 
Indeed, the probability of the second coin in \eqref{coins_prob1} is exactly the normalized importance weight of the selected particle, computed after a fresh draw of the remaining $n_j-1$ particles and of all auxiliary random variables used in the unbiased estimators of the $n_j$ weights.
A typical normalized weight is of order $\mathcal{O}(1/n_j)$, so increasing the number of particles $n_j$ decreases this probability proportionally. 
This, in turn, increases the expected cost of the two-coin procedure executed at the leaves of the DCBF tree \citep{DCBF2025}.

Ultimately, the practitioner must balance the trade-off between maintaining stable incremental importance weights (via an adequate time schedule) and controlling the computational burden of the DCBF (via the particle numbers).

\subsubsection*{DCBF specifications}

The DCBF algorithm involves two key tuning parameters: the escape probability $\beta$ in Portkey’s generalized 2-coin procedure, and the number of leaves $\ell$ in the binary tree.

The escape probability $\beta$ controls the frequency of early termination in the 2-coin routine. 
A conservative choice corresponds to selecting $\beta$ such that only a small proportion of iterations are terminated due to the $\beta$-coin, typically on the order of 1-2\%.
This limits the inflation of rejection probability while still preventing excessively costly realizations of the 2-coin algorithm. 
In practice, $\beta$ should be calibrated to achieve a substantial reduction in computational cost without materially degrading the mixing of the Markov chain.

The second specification concerns the number of leaves $\ell$ in the DCBF tree.  
The intractable probabilities $p_0$ and $p_1$ in the Barker's acceptance probability involve $J^\star$ product terms, where $J^\star = J$ for Update~2 and $J^\star = J+1$ for Update~3.
For computational efficiency, $J^\star$ should be chosen as a power of two, say $J^\star=2^K$, so that the DCBF tree is balanced with $\ell=2^{K-1}$ leaves, assigning two product terms to each leaf.
In all examples reported in Section~\ref{sec.experiments}, this specification was empirically optimal.

\section{Why rejection sampling is impractical} \label{ap:EA}

In principle, one could simulate exactly from the transition distribution of the reverse diffusion in \eqref{eq:reverse} using rejection sampling, since the Radon-Nikodym derivative $\frac{d\bar{P}}{d\bar{Q}}$ is uniformly bounded. In practice, however, this approach is infeasible because only a very loose upper bound on this derivative is available. 

While rejection sampling is efficient when the optimal bound $M_{\mathrm{opt}}$ (the supremum of the derivative) is used, employing a loose bound $M$ increases the expected number of proposals until acceptance by a factor $M/M_{\mathrm{opt}}$. In our setting, the acceptance probability deteriorates with the time $t$, decaying at least exponentially fast in $t$ and therefore vanishing as $t\to\infty$ (see Proposition~\ref{prop:RSacc}).

Specifically, assume without loss of generality that $\sup_x \pi_u(x) \le 1$. The rejection sampler for $X_t\mid X_{t+\Delta}=x$ proceeds as follows:
\begin{enumerate}
    \item Propose $y\sim \mathcal{N}(e^{(\Delta)/2} x , (e^{\Delta} - 1) I_d )$,
    \item Draw $z\sim\mathcal{N}(e^{t/2}y,(e^t-1)I_d)$ and compute $\alpha(y,z) = \pi_u(z)$,
    \item Draw $I\sim \text{Ber}(\alpha(y,z))$,
    \item If $I=1$, return $\bar{X}_t=y$, otherwise return to Step~1.
\end{enumerate}

\begin{proposition}\label{prop:RSacc}
The global acceptance probability of the rejection sampler is
\[
\mathbb{P}(\mathrm{accept})=
\mathbb{E}_{y,z}[\alpha(y,z)],
\]
where we assume without loss of generality that $\pi_u \le 1$. Moreover, if $\pi_u \in L^1(\mathbb{R}^d)$, then
\[
\mathbb{E}_{y,z}[\alpha(y,z)]=
\mathcal{O}(e^{-td/2}),
\qquad t \to \infty,
\]
and in particular $\mathbb{E}[\pi_u(W)] \to 0$ as $t \to \infty$.
\end{proposition}

\begin{proof}
Iterated expectation gives $\mathbb{P}(\mathrm{accept})=
\mathbb{E}_{y,z}[\alpha(y,z)]=
\mathbb{E}_y\left[\mathbb{E}_{z|y}[\pi_u(z)]\right]$.

Let $u=z-e^{t/2}y$. Then
\[
\mathbb{E}_{z|y}[\pi_u(z)]=\int_{\mathbb{R}^d} \pi_u(z)\tilde q_t(z|y),dz
= (2\pi (e^t-1))^{-d/2}
\int_{\mathbb{R}^d} \pi_u(u+e^{t/2}y)
\exp\left(-\frac{|u|^2}{2(e^t-1)}\right)du.
\]

By translation invariance and because $\pi_u\in L^1(\mathbb{R}^d)$,
$\ds
\int_{\mathbb{R}^d} \pi_u(u+e^{t/2}y)du
= \int_{\mathbb{R}^d} \pi_u(v)dv < \infty$.

Thus,
\[
\mathbb{E}_{z|y}[\pi_u(z)]
\le C (e^t-1)^{-d/2},
\]
with $C$ independent of $y$. Therefore,
$\ds
\mathbb{E}_{y,z}[\alpha(y,z)]
\le C (e^t-1)^{-d/2}
= \mathcal{O}(e^{-td/2}),
\;t\to\infty.
$
\end{proof}

\end{document}